\documentclass[a4paper,twoside]{article}

\newcommand{\ie}{{\em i.e., }}
\newcommand{\eg}{{\em e.g., }}

\usepackage{cite}
\usepackage{graphicx,color,subfigure}
\usepackage{graphics,url}
\usepackage[draft]{hyperref}
\usepackage{caption}
\usepackage{amsmath}
\usepackage{algorithm}
\usepackage[noend]{algpseudocode}
\usepackage[nolist]{acronym}
\usepackage{fancyvrb}
\usepackage[indentfirst=false]{quoting}
\usepackage{etoolbox}\AtBeginEnvironment{algorithmic}{\small}
\usepackage{amsthm}
\usepackage{eqnarray}

\usepackage{epsfig}
\usepackage{calc}
\usepackage{amssymb}
\usepackage{amstext}
\usepackage{amsmath}
\usepackage{amsthm}
\usepackage{multicol}
\usepackage{pslatex}
\usepackage{apalike}
\usepackage{SCITEPRESS}     

\makeatother
\algnewcommand{\LineComment}[1]{\State \(\triangleright\) #1}
\newcommand\overmat[2]{%
  \makebox[0pt][l]{$\smash{\color{white}\overbrace{\phantom{%
    \begin{matrix}#2\end{matrix}}}^{\text{\color{black}#1}}}$}#2}

\newtheorem{thm}{Theorem}

\newtheorem{lem}[thm]{Lemma}

\newtheorem{defn}[thm]{Definition}

\def\equationautorefname~#1\null{%
  (#1)\null
}

\graphicspath{{.}
{Figures/}
}

\algnewcommand\algorithmicforeach{\textbf{for each}}
\algdef{S}[FOR]{ForEach}[1]{\algorithmicforeach\ #1\ \algorithmicdo}

\quotingsetup{vskip=0pt}
\newcounter{myquote_counter}
\newsavebox{\myquotebox}
\newsavebox{\myquoterefbox}
\definecolor{quotebackcolor}{rgb}{0.9,0.9,1}
\newlength{\quotewidth}
\newlength{\quoterefindent}
\newcommand*{\scaledquotefactor}{0.95}
\newenvironment{myquote}[1]%
{\renewcommand*{\scaledquotefactor}{1.0}%
   \begin{lrbox}{\myquoterefbox}\begin{minipage}[b]{0.8\quotewidth}%
       \em #1%
   \end{minipage}%
   \end{lrbox}%
   \begin{lrbox}{\myquotebox}\begin{minipage}[b]{1.0\quotewidth} 
}%
{\end{minipage}\end{lrbox}%
\begin{center}%
 \mbox{\begin{minipage}{\quotewidth}
                        \scalebox{\scaledquotefactor}{\usebox{\myquotebox}}\\[1pt]
          \hspace*{\quoterefindent}\scalebox{\scaledquotefactor}{\usebox{\myquoterefbox}}
       \end{minipage}
      }      
\end{center}%
\par}

\begin{document}
%
\title{The Mathematical Foundations for Mapping Policies to Network Devices (\normalsize Technical Report)}

\author{\authorname{Dinesha Ranathunga\sup{1}, Matthew Roughan\sup{1}, Phil Kernick\sup{2} and Nick Falkner\sup{3}}
\affiliation{\sup{1}ACEMS, University of Adelaide, Adelaide, Australia}
\affiliation{\sup{2}CQR Consulting, Unley, Australia}
\affiliation{\sup{3}School of Computer Science, University of Adelaide, Adelaide, Australia}
\email{\{dinesha.ranathunga, matthew.roughan, nick.falkner\}@adelaide.edu.au, phil.kernick@cqr.com}
\vspace{-13mm}
}


\abstract{ A common requirement in policy specification languages 
is the ability to map policies to the underlying network devices.
Doing so, in a provably correct way, is important in a security policy context,
so administrators can be confident of the level of protection provided by the policies for their networks.
Existing policy languages allow policy composition
but lack formal semantics to allocate policy to network devices.\\
Our research tackles this from first principles: we ask how
network policies can be described at a high-level, independent of
firewall-vendor and network minutiae. 
We identify the algebraic requirements of the policy-mapping process
and propose semantic foundations to formally verify if a policy is implemented by the correct set of policy-arbiters.
We show the value of our proposed algebras in maintaining concise network-device
configurations by applying them to real-world networks. \vspace{-3mm}}

\onecolumn \maketitle \normalsize \vfill

\thispagestyle{plain}
\pagestyle{plain}




\section{\uppercase{Introduction}}
\label{sec:introduction}
\vspace{-2mm}
Managing modern-day networks is complex, tedious and error-prone. These networks are comprised of  a wide variety of devices, from switches and routers to middle-boxes  like firewalls, intrusion detection systems and load balancers. Configuring these heterogeneous devices with their potentially complex policies manually, box-by-box is impractical.

A better approach is a {\em top-down configuration} where device configurations are derived from a high-level user specification. 
Such specifications raise the level of abstraction of network policies above device and vendor-oriented APIs (\eg Cisco CLI, OpenFlow). 
Doing so, provides {\em a single source of truth}, so, security administrators for instance, can easily determine who gets in and who doesn't \cite{howe1996}. 

In recent years, several research groups have proposed high-level policy specification languages for both \ac{SDN} \cite{soule2014,reich2013,foster2010,prakash2015} 
and traditional networks \cite{bartal1999,guttman2005}.
A common requirement in these languages is the ability to map the abstract policies to the underlying physical network. Doing so accurately, is essential to:
\begin{itemize}
\item enforce policy correctly;
\vspace{-1mm}
\item track policy changes per network device, so, redundant policy updates can be avoided; 
\end{itemize}
\smallskip
But, mapping policies to the physical network devices has challenges:

\begin{itemize}
 
\item policy needs to be {\em decoupled} from the network (\ie free of vendor and network-centric minutiae). A policy shouldn't need to change when
the device vendor changes, or every time IP address changes occur;
\vspace{-1mm}
\item the mapping must be {\em formally verifiable}. Precise and unambiguous mathematical semantics would eliminate wishful thinking pitfalls in deploying policies to networks. These semantics give, for instance, security administrators assurance that their intended policies are enforced by the correct
firewalls, rendering the expected security outcome.
\vspace{-1mm}
\item policies can have complex semantics including node and link properties; and
\vspace{-1mm}
\item the mapping should be as efficient as possible.
\end{itemize}

The solution we propose provides a generic framework to map policies to network devices algebraically. 
We illustrate it's use by considering security policies, 
because incorrect deployment of these policies in domains such as \ac{SCADA} networks can result in
catastrophic outcomes including the loss of human lives!
However, the principles involved, can easily be extended to policies involving traffic measurement, QoS, load balancing and so on. 

In developing our algebras, we have derived the properties and constraints of sequentially- and parallelly-composed policies for various policy contexts.
These policy-composition semantics must be preserved, when mapping policy to devices, to ensure correct deployment.
Using an algebraic framework also makes the policy-mapping process efficient.

We will describe the implementation of our proposed algebras and demonstrate their use in deploying security policies to real SCADA networks.
Particularly, we will show it's value in maintaining a clear, concise set of firewall configurations.
Moreover, our approach allows administrators to conduct ``what if'' analysis by changing policy and/or network topology 
and observe their effect on the network devices required to implement the changes.
\vspace{-6mm}

\section{\uppercase{Background and Related work}}
\label{sec:background}
\vspace{-5mm}
\smallskip
\begin{myquote}{Bertrand Russel}
``The advantages of implicit definition over construction are roughly those of theft
over honest toil.''
\end{myquote}

The quote is salient because network installations are commonly built from {\em bottom-up}, 
\ie a network-device is purchased, and configurations written. The policy is the result of the configuration, which is the
consequence of a purchasing decision. So, the policy is implicitly defined as a result of many small decisions
that interact in complex ways. Instead, best-practice guides \cite{byres2005,isa2007} suggest designing the policy first, and only then
determining how to implement it. 

Solutions that employ such a {\em top-down} network configuration approach have been proposed in both 
\ac{SDN} \cite{soule2014,anderson2014,prakash2015} and traditional networks \cite{bartal1999,cisco2014b}. 
They allow creation and maintenance of a {\em single high-level network-wide policy}
(\ie source of truth), so network administrators can easily determine, for instance,
who gets in and who doesn't. These high-level policies should be decoupled from
network and device-vendor intricacies, so they capture the {\em policy intent} and not the {\em policy implementation}.

Capturing intent has several benefits: it allows policy to be distinguished from network to assist with change management;
it enables accurate comparison of organisational policies to industry best-practices to evaluate compliance; and
it allows policy semantics to be expressed without network intricacies like IP addresses.
But, most research towards high-level policy languages \cite{bartal1999,cisco2014b,soule2014}, still requires hostnames and/or IP addresses to be
explicitly input to the policy specification, to implement policy on a network instance.

If the high-level policy definition is built on {\em formal mathematical constructions}, then
there are no implicit properties and it provides a truly sound foundation for everything that follows. 
The formalism would allow construction of complex and flexible policies and support reasoning about the policies.
For instance, we could precisely compare a defined policy with industry best-practices in \cite{isa2007,byres2005} for compliance and reduce 
network vulnerability to cyber attacks \cite{ranathunga2015P}. 

It is equally important to map policy to devices using a formal approach.
For instance, we can be confident of the blanket of protection
provided for our network {\em if and only if} we could prove that an intended security policy is implemented by the correct set of network firewalls.
Existing top-down configuration languages \cite{bartal1999,prakash2015,cisco2014b} allow creation of network-wide high-level policies, but
lack means to allocate policy to network devices in provably correct way.

In our related work \cite{ranathunga2015V}, we have developed a high-level security policy specification 
that is network and device-vendor independent
to perform {\em top-down} configuration of network firewalls.
The language semantics allow these high-level policies to be
easily understood by humans. 
In this paper, we investigate the underlying requirements for
mapping high-level policy to network devices.
\vspace{-6mm}

\section{\uppercase{Abstract policies}}
\label{sec:abstractions}
\vspace{-2mm}
Abstractions are key to constructing high-level policies.
A good policy abstraction should capture the underlying 
concepts naturally and concisely. For instance, a policy may be arbitrated using one or more network devices.
A good abstraction should decouple {\em what} is arbitrated from {\em how} it is arbitrated.

A simple abstraction that is commonly used to decouple policy from the network is {\em (endpoint-group, edge)} \cite{ranathunga2015V,bartal1999,soule2014,prakash2015}.
An {\em endpoint} is a grouping of elements based on common physical or logical properties- \eg a subnet, a user-group, a collection of servers, {\em etc.}
An {\em endpoint} is also the smallest unit of abstraction a policy can be applied to.
An {\em edge} specifies the relationship between the endpoints, \ie it describes what communication is allowed between the endpoints.
An edge typically consists of a Boolean predicate (\ie classifier) that matches traffic packet header details to perform some action.
An edge from endpoint $S1$ to $S2$ with predicate $R$ means that a correct implementation should act in some specified way toward traffic from all members (\eg hosts) in $S1$ to those in $S2$ satisfying $R$.

Consider the example network in Figure \autoref{fig:network}, our high-level policy might be -- ``we want to measure all HTTPS flows from $S1$ to $S4$".
This intent can be expressed using the (endpoint-group, edge) abstraction as \verb|p:=S1| $\rightarrow$ \verb|S4 : collect https|, where $p \in \mathcal{P}$ is an element of the possible space $\mathcal{P}$ of policies.

We now want to map the policy to enforcing devices (\ie arbiters). It looks naively easy. Just implement the policy on arbiter $F$. However, the problem is in general more complicated.
For example, consider policy from $S1$ to $S2$. If arbiter $A$ captures flows described by policy $p_A$ and likewise $p_B$, then the combination is
$p_{A \cap B}$ because the policy expresses that flows be collected but load balancing across the two paths could result in either being used by a given packet.

For another example, consider policy from $S2$ to $S3$. If arbiters $C$ and $D$ capture flows described by policy $p_C$ and $p_D$ respectively, then
the combination would yield their {\em union} because both arbiters in the path are used for flow collection. 
So, when mapping policies to network devices, they must be composed in a way that preserve the topology-specific semantics.

We will, in general, denote the policy composition resulting from parallel routes as $p_A \oplus p_B$ and that resulting from serial routes as $p_A \otimes p_B$.
In this example, $p_A \oplus p_B=p_{A \cap B}$ and $p_C \otimes p_D=p_{C \cup D}$.
Note though, the meaning of $\oplus$, $\otimes$ are policy-context dependent.

The problem of mapping our high-level policy $p$ (from $S1$ to $S4$) to network devices, may even be more complicated than simply implementing the policy on $F$.
What happens if $F$ drops out? Surely we would still want to collect HTTPS flows that traverse the redundant paths from $S1$ and $S4$ (\eg path via $S2$ and $S3$). So, $p$ must also be mapped to the policy arbiters along the path $S1-S2-S3-S4$, to cater for this redundancy. 
But then, should we map the policy to all of these arbiters or a subset of them?

Given the parallel routes between $S1$ and $S2$ it is easy to see that we need to map $p$ to both $A$ and $B$ to preserve the semantics of $\oplus$.
But, when arbiters are in series (\eg $C,D$) we have the choice of mapping policy to both or just one to preserve the semantics of $\otimes$.
Mapping to both may be unnecessary and inefficient and arbiters may have limits on their capabilities.

\begin{figure}[t!]
\captionsetup{aboveskip=8pt}
\centering
\subfigure[Network topology.] 
{
	\includegraphics[scale=0.32]{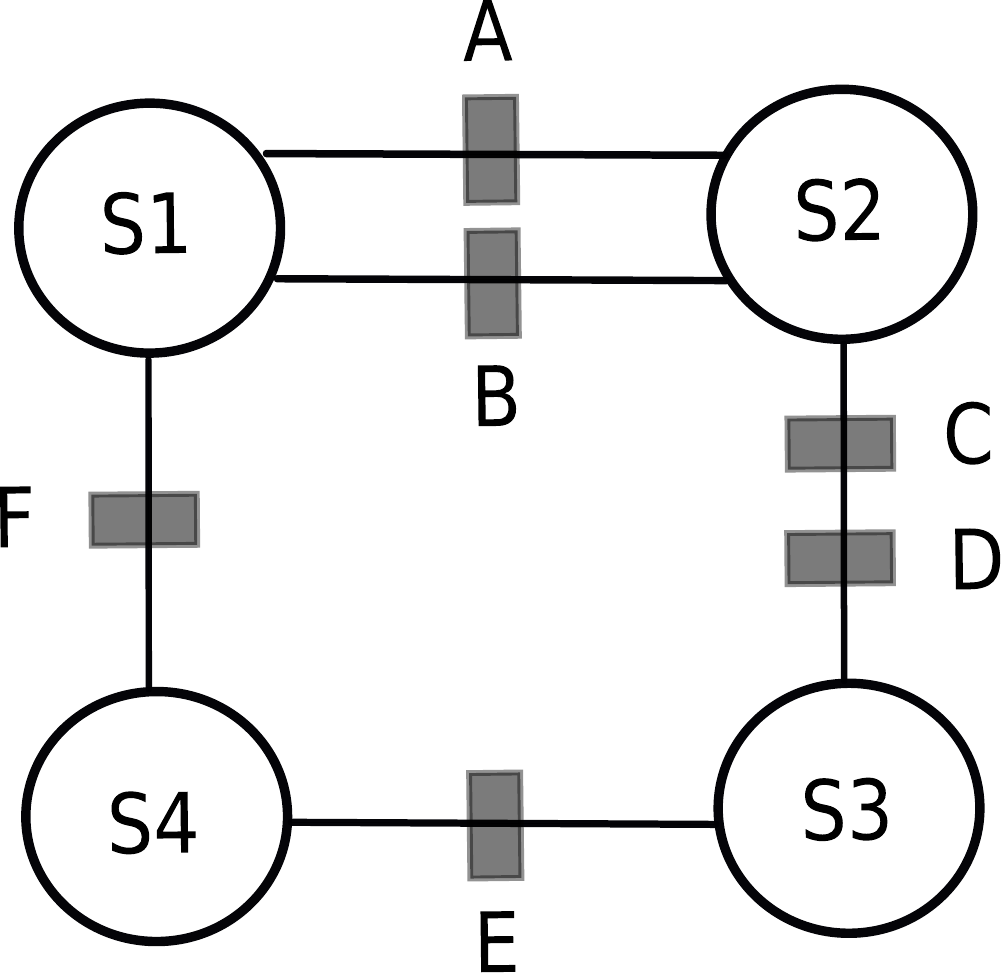}
	\label{fig:network}	
}
\hspace{0.1cm}
\subfigure[Policy graph of (a).] 
{
	\includegraphics[scale=0.32]{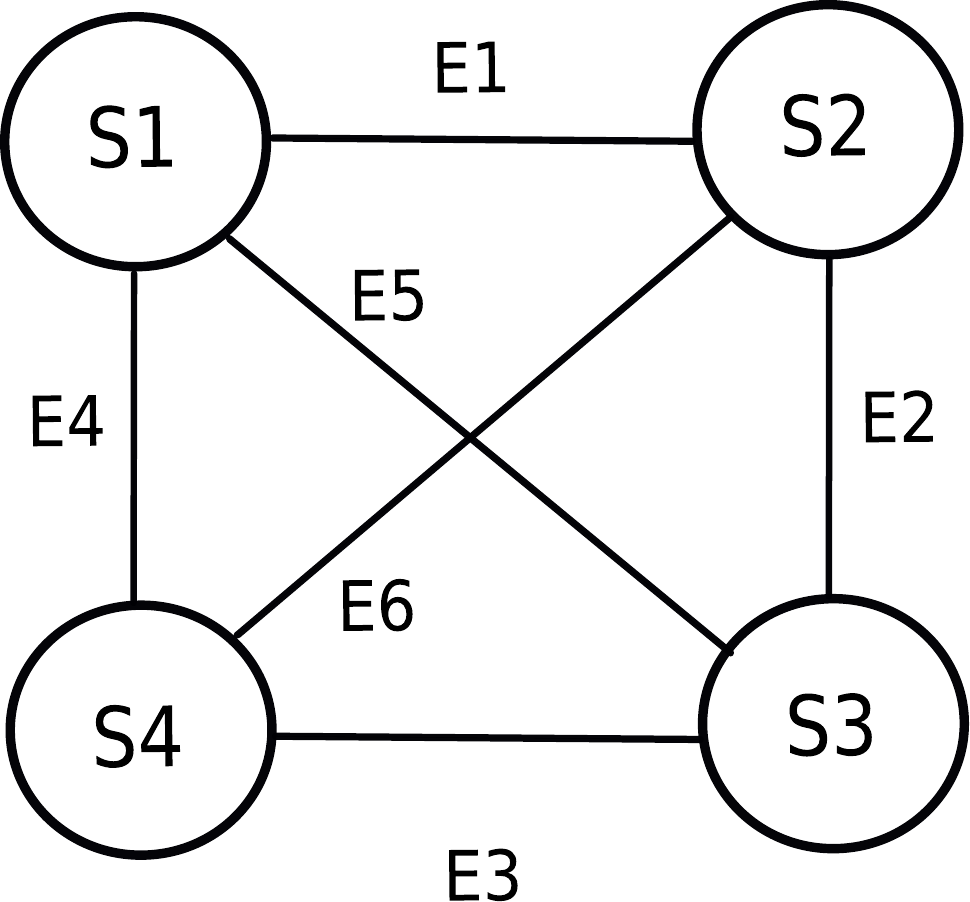}
	\label{fig:abs_model}
}
\caption{Example network consisting of four endpoints ($S1-S4$) and six policy arbiters ${A,B,...,F}$ (a). The corresponding policy graph is shown in (b) with these endpoints and 
policy edges ($E1,E2,...,E6$) between them.}
\label{fig:network_abs}
\end{figure}

\begin{figure}[t!]
\captionsetup{aboveskip=8pt}
\centering
\subfigure[Network of 3 endpoints connected to a single device.] 
{
	\includegraphics[scale=0.175]{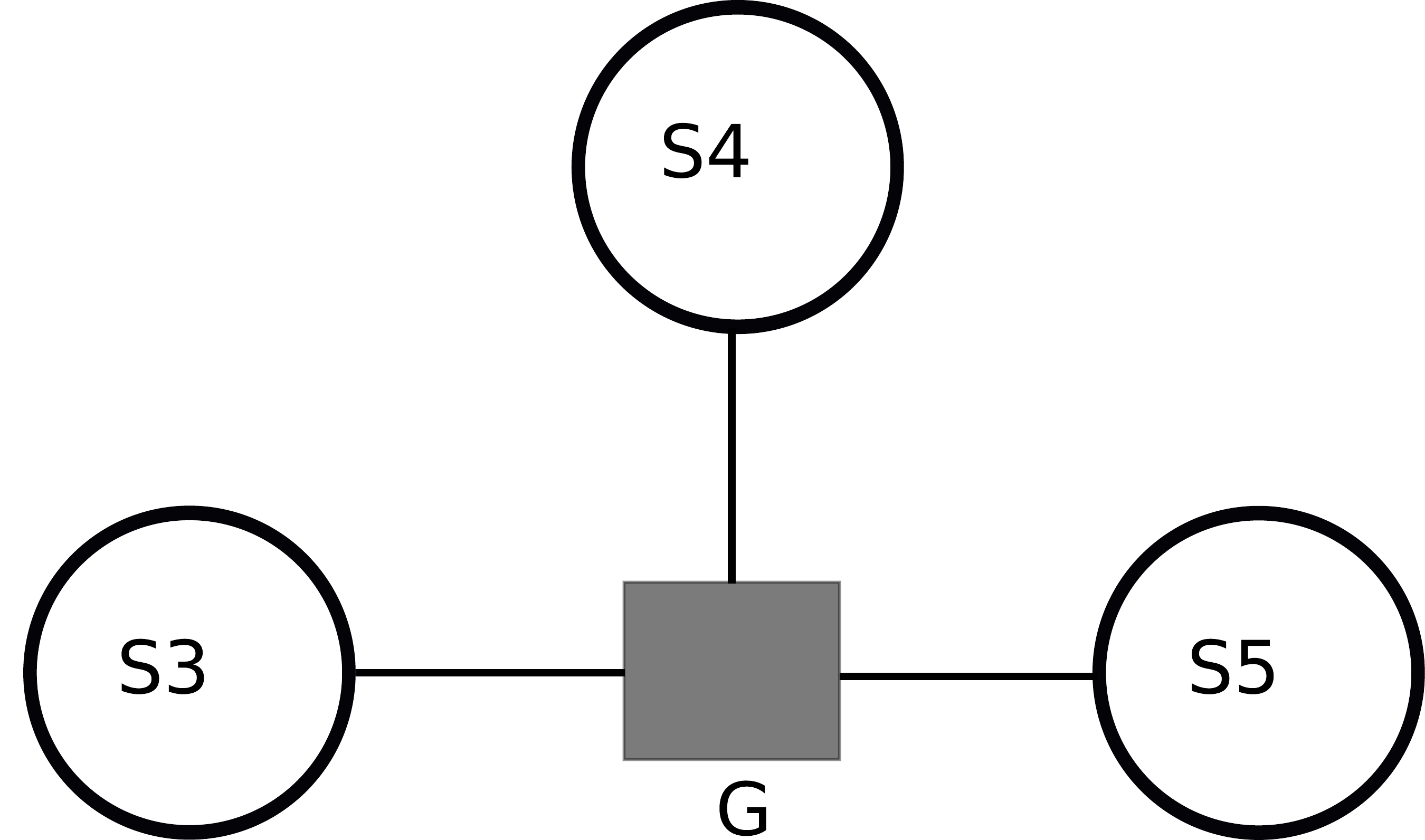}
	\label{fig:star}	
}
\hspace{0.3cm}
\subfigure[Policy hyper-graph model of (a).] 
{
	\includegraphics[scale=0.215]{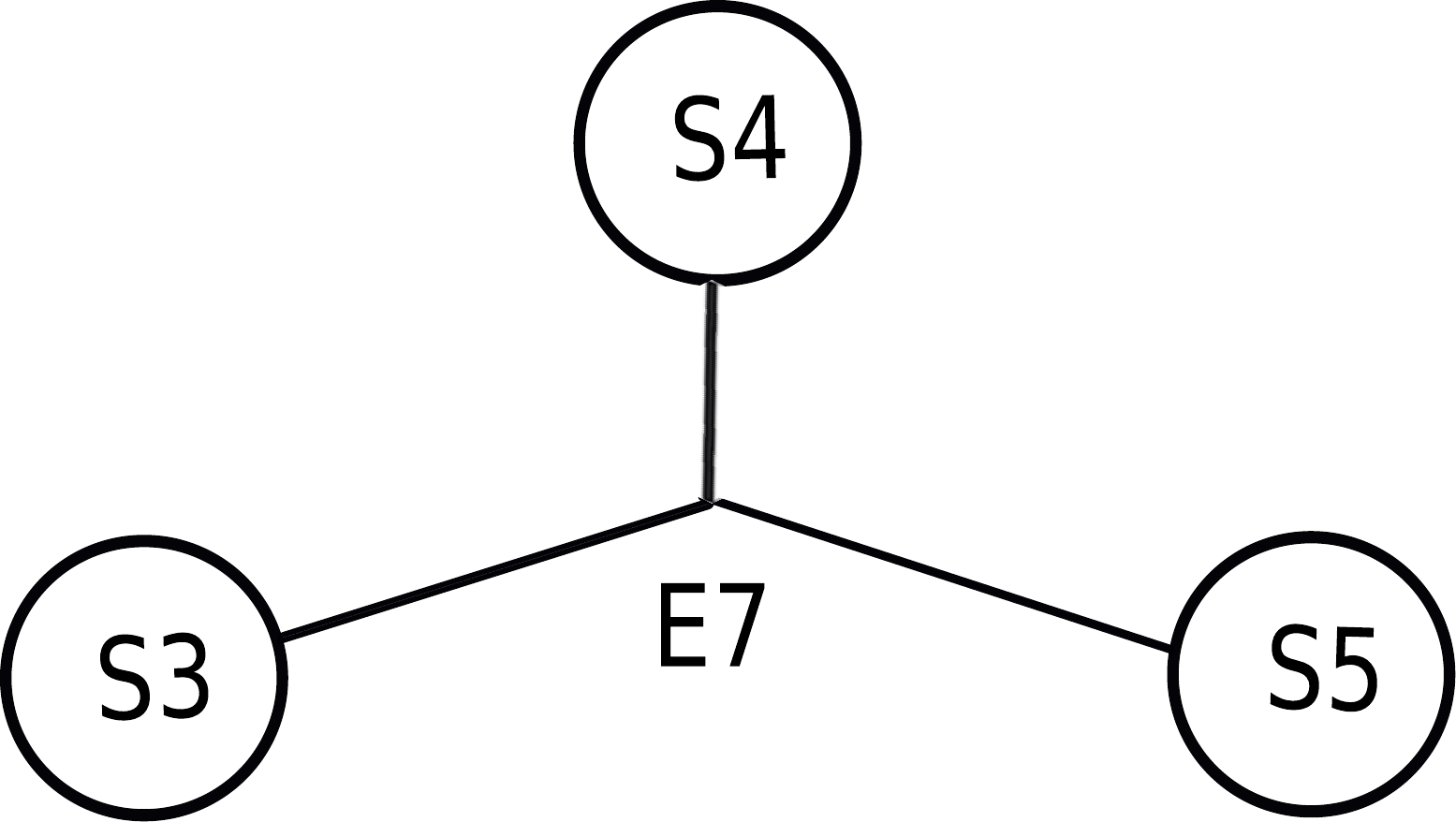}
	\label{fig:hyper}
}
\hspace{0.15cm}
\subfigure[Policy simple-graph model of (a).] 
{
	\includegraphics[scale=0.215]{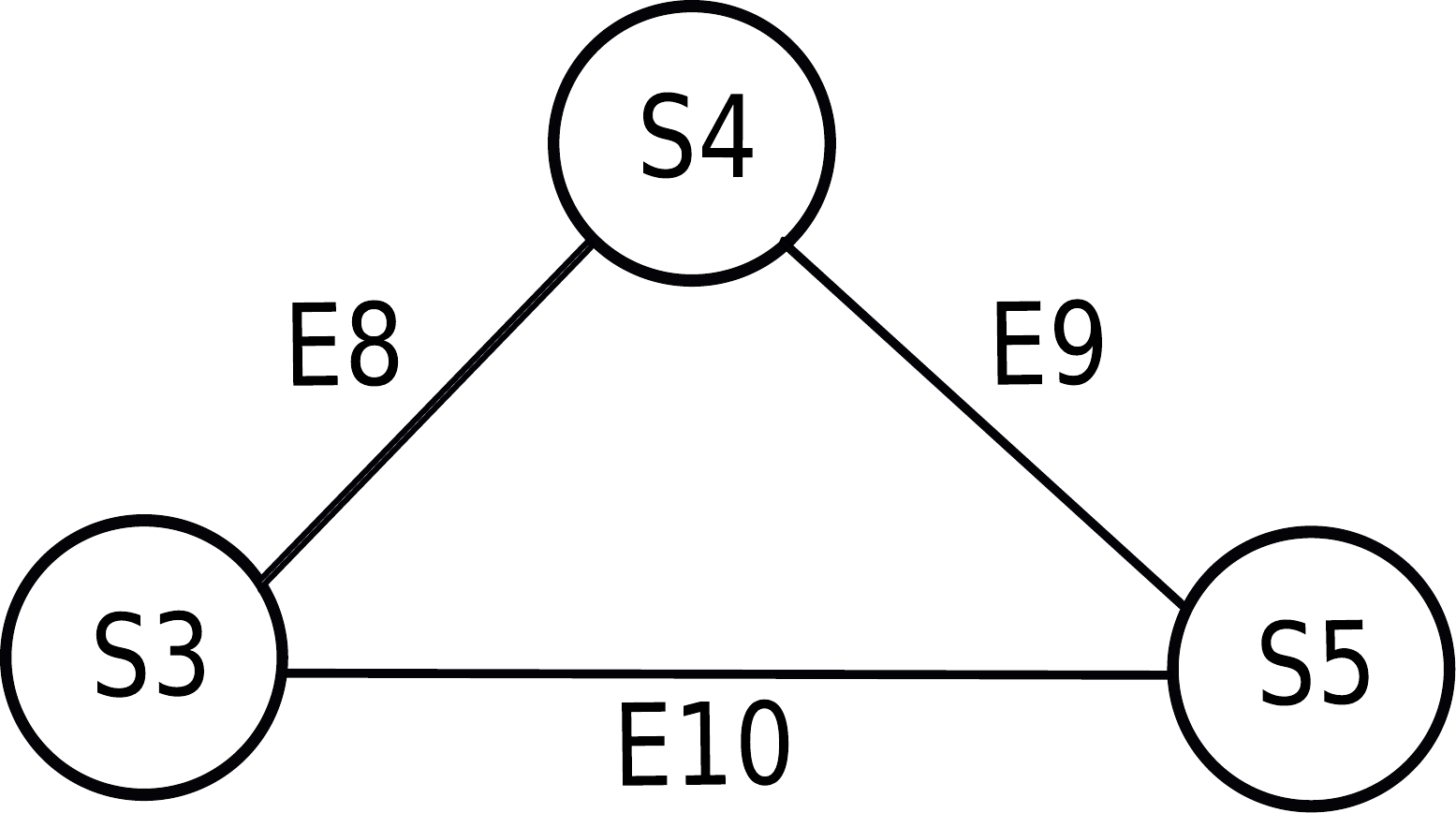}
	\label{fig:simple}
}
\caption{Policy-graph definition alternatives.}
\label{fig:alternat}
\end{figure}

The policy arbiters may not always be in series or parallel. For instance, consider the star-topology involving the endpoints $S3,S4$ and $S5$ in Figure \autoref{fig:star}. The question of how to map policy to topology gets interesting in this case. We could loosely map policy to arbiter $G$ using a {\em hyper-edge} $E7$ (Figure \autoref{fig:hyper}).
But, doing so adds to the complexity of the policy-graph and decreases it's precision.
For instance, it allows a $1:n$ relationship between hyper-edges and policies, so, we must track the policy of {\em every endpoint-pair} in a hyper-edge
to correctly map it's policy to network devices. Also, a hyper-edge could potentially map to a large portion of the network (\eg when interconnected to many endpoints), so, it's policy can be complex, difficult to understand and debug.

We can overcome these problems by using a simple-edge mapping instead (Figure \autoref{fig:simple}).
By doing so, we also gain significant benefits.
For one, the policy graph is simplified and it's precision is increased-- the $1:1$ relationship between edges and policies now mean we only need to track a single policy for each edge, to map policy accurately to devices. For another, our simplified abstraction {\em (endpoint-pair, edge)} 
behaves as a network-slice \cite{gutz2012}-- a piece of the network that can be programmed independently from the rest of the network -- so, debugging policy errors is easy. Moreover, users from distinct policy subdomains (\eg Engineering, Corporate Administration) can manage the network-slices (\ie policy-edges) applicable to their domain. 

We may also need to consider paths that consist of one more intermediate endpoints,
when implementing policy between an endpoint pair.
For instance, consider implementing a security policy $f:=$\verb|S1| $\rightarrow$ \verb|S3: http| in the network in Figure \autoref{fig:network}.
Once implemented on the arbiters (\ie firewalls), the policy should allow HTTP traffic flow from \verb|S1| to \verb|S3|.
But, the policy can only be deployed correctly if endpoints \verb|S2,S4| can route or forward HTTP traffic.

This ability of an endpoint to route or forward traffic can be restrictive.
For instance, in a security policy context, an endpoint can group systems (\eg hosts) with similar security requirements \cite{isa2007}.
So, high-risk systems can be grouped in to one endpoint and only {\em safe} traffic that originate and/or terminate at the endpoint is permitted by the security policy. 
Enabling this endpoint to transit-traffic could potentially expose the high-risk systems within to cyber attacks.
So, an endpoint's {\em traffic transitivity} capability must be considered in constructing valid device-paths between endpoints and
captured explicitly in the mapping process.

The intent of our high-level policy $p$ is to collect HTTPS flow data from $S1$ to $S4$.
Similarly, traffic measurement policies can be defined between any pair of endpoints in the network, so the corresponding
(endpoint-pair,edge) tuples collectively construct a {\em policy-graph} (Figure \autoref{fig:abs_model}).
Each (logical) edge in this graph maps to a physical network path comprising of policy arbiters and zero or more intermediate endpoints.
For instance, policy-edge $E4$ in Figure \autoref{fig:abs_model} implements our policy $p$.


So far, we have described several key requirements of a policy-to-device mapping process.
We next illustrate these using more detailed examples involving a typical set of policies found in practice.
The examples also show the use of our {\em (endpoint-pair, edge)} abstraction in high-level policy description.
\vspace{-3mm}
\subsection{Quality of Service (QoS) Policies}
\vspace{-3mm}
QoS policies can be used, for instance, to provide bandwidth guarantees for certain types of traffic.
The policy arbiters in this context, would commonly be QoS capable routers or switches.
Imagine we need to provision a minimum bandwidth of 100MB/s for HTTP traffic flow from $S1$ to $S4$ in Figure \autoref{fig:network}.
The high-level policy intent can be expressed using the endpoint-pair $S1$, $S4$ and an edge ($E4$) between them: \verb|min(S1| $\rightarrow$ \verb|S4 : http, 100MB/s)|.
In reality, similar QoS policies between any endpoint-pair can be expressed using the same policy-graph in Figure \autoref{fig:abs_model}.

Parallel and serial (QoS-device) topologies also have an impact on the end QoS policy.
Consider the parallel topology between $S1$ and $S2$ in Figure \autoref{fig:network}, if we assume $p_A$ and $p_B$ provide bandwidth guarantees
of $B1$ and $B2$ respectively, then with load balancing, the resultant QoS policy ($p_R$) 
can provide a total bandwidth guarantee $p_R=p_A \oplus p_B = sum(B1,B2)$.

When the devices are in series (\eg topology between $S2$ and $S3$ in Figure \autoref{fig:network}), the resultant policy provides a bandwidth guarantee given by $p_C \otimes p_D=min(B3,B4)$ where $B3,B4$ are the bandwidth guarantees provided by $p_C$ and $p_D$.

This example highlights how policies can be composed using the generic semantics $\oplus$ and $\otimes$ in different policy contexts.
The actual meaning of these operators is policy-type dependant.
For instance, in the traffic measurement policy example, $\otimes$ represented {\em union} while here
it has meaning of {\em minimum}.

We can see that endpoint traffic-transitivity also impacts the ability to implement a QoS policy correctly in the network.
For instance, consider policy \verb|min(S1| $\rightarrow$ \verb|S3 : http, 50MB/s)|, the required bandwidth
guarantee can only be provided if $S2,S4$ routes or forwards traffic to $S3$.
\vspace{-3mm}
\subsection{Security Policies}
\label{sec:sec}
\vspace{-3mm}
The security policies we consider here are access-control policies in a network.
The policy arbiters in this context, would be traffic filtering devices (\eg firewalls, SDN switches).
The endpoints could be zones or user groups.
Imagine we want to enable only SSH traffic from $S1$ to $S4$ (Figure \autoref{fig:network}). The high-level policy can be expressed using the endpoint-pair $S1$, $S4$ and an edge between them as \verb|S1| $\rightarrow$ \verb|S4: ssh| with an
implicit \verb|deny-all| in place.

When the traffic filtering devices are in parallel (\eg topology between $S1$ and $S2$ in Figure \autoref{fig:network}), the resultant security policy ($p_R$), has meaning- {\em all packets that can {\em possibly} be allowed through} -- and is the {\em union} of the packet sets allowed through by the individual devices. We take union conservatively because intrusions and attacks are usually carried out through permitted traffic.
So, if $p_A$ and $p_B$ allows packet sets $Q$ and $T$ respectively, then $p_R=p_A \oplus p_B = p_{Q \cup T}$.

When the devices are in series (\eg topology between $S2$ and $S3$), the resultant policy permits the {\em intersection} of the packet sets $V,W$ allowed by the policies of $C$, $D$ respectively -\ie  $p_R=p_C \otimes p_D = p_{V \cap W}$. 

With security policies, it is often useful to have endpoints that group systems with similar security requirements.
Doing so, allows to define a {\em single policy} for all members of an endpoint-- a clean and concise abstraction of network security.
But, a generic (endpoint-pair, edge) abstraction cannot capture such network-security specific policy concepts precisely.

So, we need concrete definitions for what an endpoint and an edge means for each policy context
to capture policy-specific intricacies.
We will show later how to define these concretely for security policies.
\vspace{-3mm}
\subsection{Traffic Measurement Policies}
\label{sec:measure}
\vspace{-3mm}
Traffic measurement policies can be used to implement, for instance, NetFlow (v9) on the example network shown in Figure \autoref{fig:network}.
The policy arbiters ${A,B,...,F}$ here, would be NetFlow capable devices (usually routers or switches).
The endpoints could be subnets or zones in the network.

When we considered the parallel routes between $S1$ and $S2$,
the resultant policy ($p_R$) with meaning--
{\em the flows that are guaranteed to be captured by the topology (without sampling)},
constitutes of the {\em intersection} of the packet sets of $A$ and $B$.
So, if $p_A$ and $p_B$ capture packet sets $Q$ and $T$ respectively, then $p_R=p_A \oplus p_B = p_{Q \cap T}$.

We also showed that when the devices are in series (\eg topology between $S2$ and $S3$ in Figure \autoref{fig:network}), the resultant policy captures the {\em union} of the packet sets $V, W$ captured by the policy of the devices $C$ and $D$ respectively, \ie  $p_R=p_C \otimes p_D = p_{V \cup W}$. 

\smallskip
\smallskip
\noindent \textbf{\em Policy Mapping Vs Routing:} We can see that in mapping policies to network devices, our discussion above relates to network paths or routes.
This is no accident. The problem we aim to solve has many parallels with routing. But here,
\begin{itemize}
\item{we often look for all paths, not just best paths.}
\vspace{-1mm}
\item{we do not need to employ decentralised protocols. A logically-centralised implementation can offer a robust and efficient solution
as demonstrated in SDN \cite{anderson2014,foster2010,soule2014,prakash2015}.}
\end{itemize}

The aim in routing is to determine the path that optimises a given path-metric (\eg shortest-path routing finds
a path between endpoints with a minimum distance). Our target problem is different: we need to determine the set of arbiters in a network a given policy should be implemented on. So, constructing all feasible paths is crucial because, for instance, a security policy between two endpoints can only be correctly implemented (\eg to block all Telnet traffic) if all redundant paths between them are taken into account.

We also expect a relatively low number of endpoints (\eg in \autoref{sec:example} we will see this number is typically $< 25$ in a SCADA network),
through network grouping. As we will show in \autoref{sec:imp}, a low endpoint-count makes our policy-mapping algorithm computationally tractable.
So, a logically-centralised implementation can be considered. But, in routing, the number of endpoints (\eg individual gateways) can be high for a large network, so, distributed means to computing a solution is essential (requiring de-centralised protocols like OSPF, BGP) . Our policy-mapping algorithm can also be implemented distributedly if necessary.

Algebras have also been proposed in meta-routing \cite{griffin2013}, to provably choose paths that optimise various path-metrics.
But, these algebras {\em do not support specification of node properties}
such as traffic transitivity. Specification of link properties is possible and we could translate node properties to
link properties, but, the workaround is not elegant.
A better approach is to have an algebra that capture node properties explicitly.

\smallskip
\noindent We described in \autoref{sec:sec}, the need to define endpoints and edges concretely for each policy context.
By doing so, we can incorporate the intricacies associated with each policy type.
We illustrate the idea next using security policies and their concretised (endpoint-pair, edge) abstraction- {\em the Zone-Conduit model}.

\smallskip
\noindent \textbf{\em The Zone-Conduit Model}:
Lack of internal-network segmentation is a key contributor to the fast propagation of threats and attacks in a network \cite{byres2005}.
So, ANSI/ISA have proposed the {\em Zone-Conduit model} as a way of segmenting and isolating the various sub-systems in a SCADA network \cite{isa2007}.

A \textit{zone} is a logical or physical grouping of an organisation's
systems with similar security requirements so that a {\em single policy}
can be defined for all zone members. A \textit{conduit}
provides the secure communication path between two zones, enforcing
the policy between them \cite{isa2007}.
A conduit could consist of
multiple links and firewalls but, logically, is a single connector.

It is easy to see that the Zone-Conduit model is a concrete instance of the (endpoint-pair, edge) abstraction for high-level security policy specification.
The zones and conduits are concrete definitions of endpoints and edges.
In defining them, important network-security characteristics such as {\em a single zone-policy}, are captured concisely.

The Zone-Conduit model abstracts {\em how} security policy is enforced, so users can focus on
{\em what} policy to enforce.
The model intuitively decouples policy from topology, so, high-level security policies 
can be described free of network and vendor intricacies. 

The ISA Zone-Conduit model in its original description lacks precision for policy specification. 
We use the extensions proposed in
\cite{ranathunga2015} to increase its precision, \eg we add Firewall-Zones to specify firewall-management policies.

\begin{figure*}[t!]
\captionsetup{aboveskip=8pt}
\centering
\subfigure[Zone-Firewall model.] 
{
	\includegraphics[scale=0.36]{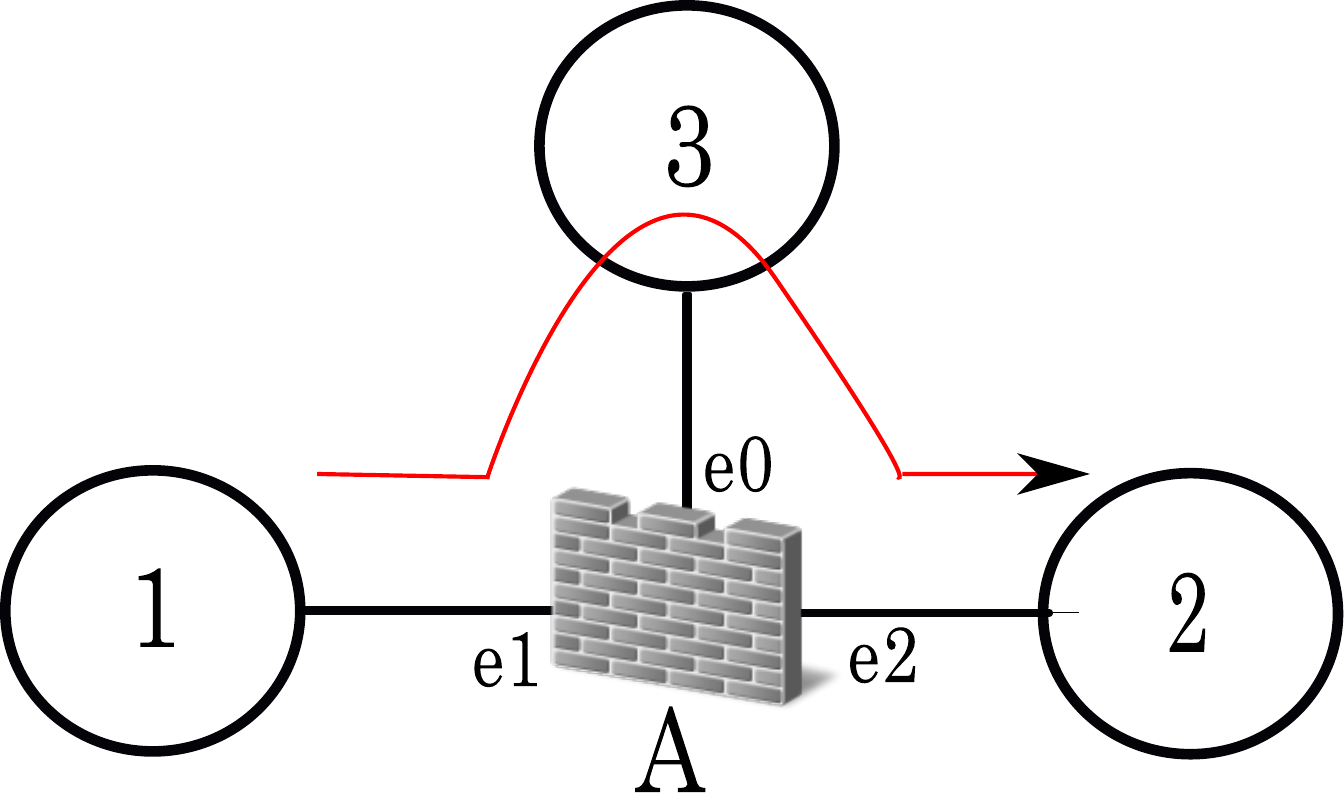}
	\label{fig:vc1}
}
\hspace{0.8cm}
\subfigure[Firewall-path model of (a).] 
{
	\includegraphics[scale=0.37]{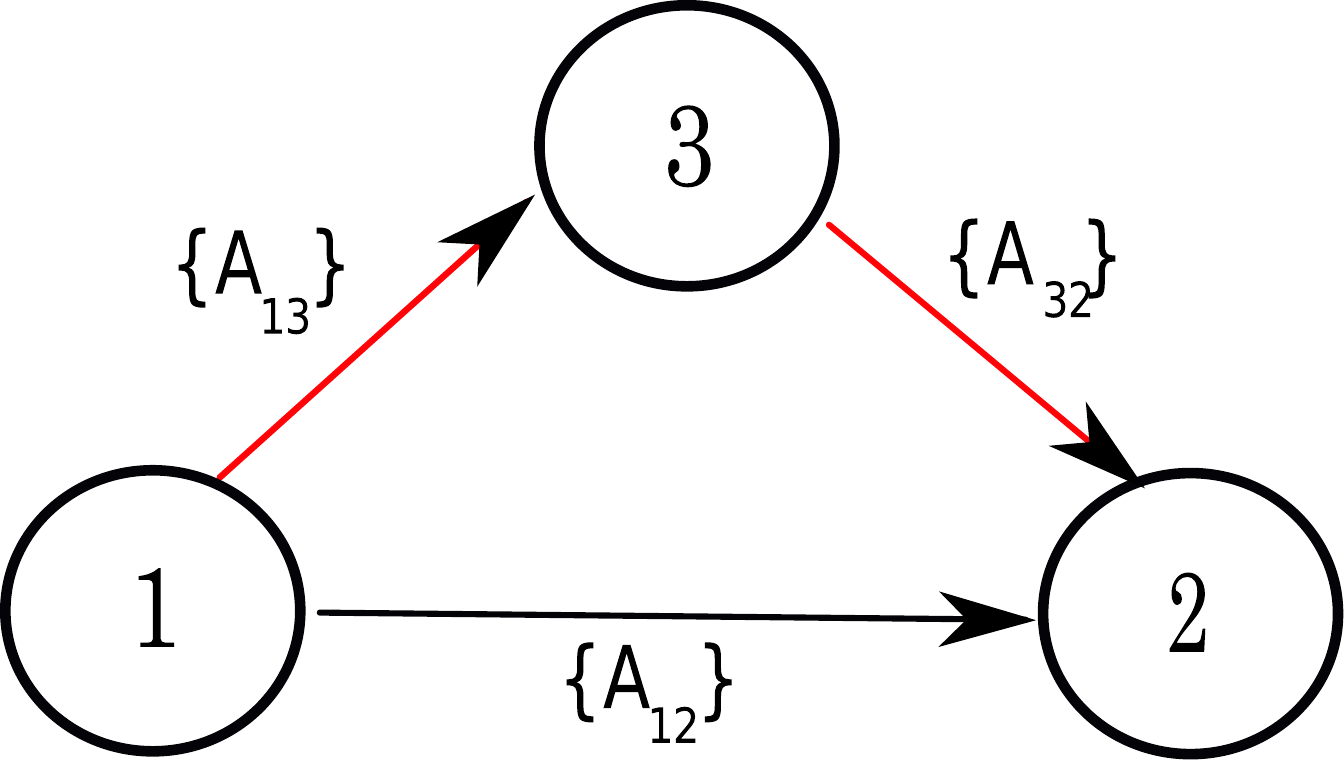}
	\label{fig:vcf1}
}
\caption{Example of traffic packets traversing a single firewall interface twice. As shown by the path in red, if traffic from zone 1 to 2 is allowed to transit zone 3, 
firewall interface e0 is traversed twice. In this situation, the direct path from zone 1 to 2 is more efficient. So, we deem firewall path $A_{13}A_{32}$ invalid (\ie =$\phi$), since
$h(A_{13})=h(A_{32})=A$, where $h: F \rightarrow W$.}
\label{fig:fwloop}
\end{figure*}

\smallskip
\noindent \textbf{\em Zone-Conduit Policies}:
A conduit policy is essentially an ordered set of rules $[p_1,p_2,...,p_n]$ that act on packet sequences to accept, deny, or in some cases, modify them.
In our related work \cite{ranathunga2015V,ranathunga2015P}, we have identified properties and constraint required in a Zone-Conduit based policy description to make these rules vendor and device (\ie implementation) independent
and have rule-order independent semantics. 

To summarise, we adopt a security whitelisting model, \ie we restrict policies to express positive abilities\footnote{Refers to the ability to initiate or accept a traffic service.} and {\em deny all} inter-zone flows that are not explicitly allowed \cite{ranathunga2015V}. 
Doing so, renders the rule order irrelevant in a policy, 
allowing the policy to be converted consistently to different vendor firewalls. So, the underlying vendor can change without requiring policy alterations. Policy makers can also add or remove policy rules oblivious to the ordering of the rules. 
By being explicit, we also prevent services being accidentally enabled implicitly. 

\smallskip
\noindent \textbf{\em Directed- vs Undirected-Conduit Policy}:
The policy on an undirected-conduit can be expressed using two {\em directed-conduit} policies.
Directed-conduits are important in understanding how policy should
be implemented on device interfaces. For instance, consider implementing
a security policy $p:=Z1 \rightarrow Z2: https$ on the network shown in Figure \autoref{fig:vc1}.
An undirected-conduit can only map the policy to firewall interfaces $e1,e2$.
But, to implement the policy correctly, we additionally need to know if it should be implemented 
{\em inbound} or {\em outbound} on these interfaces. A directed-conduit provides this directionality.

However, analysis using directed-conduits can also lead to problems. For instance,
it seems feasible to implement our policy $p$ along the path highlighted in red in Figure \autoref{fig:vcf1}.
But further inspection of Figure \autoref{fig:vc1}, reveals that traffic packets traversing this path 
would require to traverse firewall interface $e0$ twice.


We invalidate such directed-conduit paths via a mapping
$h: F \rightarrow W$ where $F=\{directed \ firewalls\}$ and $W=\{firewalls\}$.
A directed-firewall in the Zone-Conduit graph $G=(Z,C)$ is defined as
\vspace{-2mm}
\begin{defn}[Directed firewall]
  A {\em directed firewall} $t_{ij}$ is a firewall $t \in W$ that filters traffic on directed-conduit $(i,j) \in C$.
  \label{def:directed_firewall1}
\end{defn}
\vspace{-2mm}
\noindent Then we can check if the directed-firewalls in a path map to the same physical firewall
(\ie $h(A_{13})=h(A_{31})=A$) and deem that path invalid.

\vspace{-7mm}

\section{\uppercase{Mapping Algebra}}
\label{sec:algebra}
\vspace{-3mm}
We now outline our proposed algebras to map high-level policy to network devices.
We start by making the distinction between {\em primary} and {\em secondary}
policy-edges, and then later show how the latter can be derived algebraically using the former.
\vspace{-2mm}
\subsection{Firewall-path Construction}
\label{sec:encoding}
\vspace{-3mm}
As we described earlier, a policy-edge enforces policy between two endpoints in a network.
A policy-edge can be classified as a {\em primary-} or a {\em secondary-edge} depending what
arbiters are implemented within.
A primary-edge can enforce policy {\em only} using policy arbiters.
A secondary-edge enforces policy using both arbiters and one or more additional endpoints.

We demonstrate the idea using security policies and the simple Zone-Conduit graph $G=(Z,C)$ in Figure \autoref{fig:vc50}. 
The firewall composition of each primary-conduit in the model is shown (Figure \autoref{fig:vcf51}).
The example secondary-conduit $C_{13}$ filters traffic flow from zone \verb|1| to \verb|3|, using several directed-firewalls
and transit zones \verb|2|, \verb|4|.
The set of all directed-conduits are given by $DC=\{C_{ij} \ | \ (i,j) \in C \}$.

A policy-graph is essentially an automation that moves traffic packets from 
one endpoint to another using policy-edges. So, regular expressions (the natural language of finite automata),
can capture the packet-processing behaviour of this model; a path encoding is a concatenation of directed-devices (\ie policy-arbitration steps $pq$) 
and a set of paths is encoded as a union of paths.
Past work \cite{anderson2014} has shown, all single-path encodings stem from the Kleene star operator (*) on the set of directed-devices.

In a security policy context, the automation means that a single firewall-path encoding is a concatenation of directed-firewalls.
Each path depicts a sequence of traffic filtering steps and is an element of $F^*$.

\begin{figure*}[t!]
\captionsetup{aboveskip=8pt}
\centering
\subfigure[Zone-Conduit model with primary and secondary conduits.] 
{
	\includegraphics[scale=0.4]{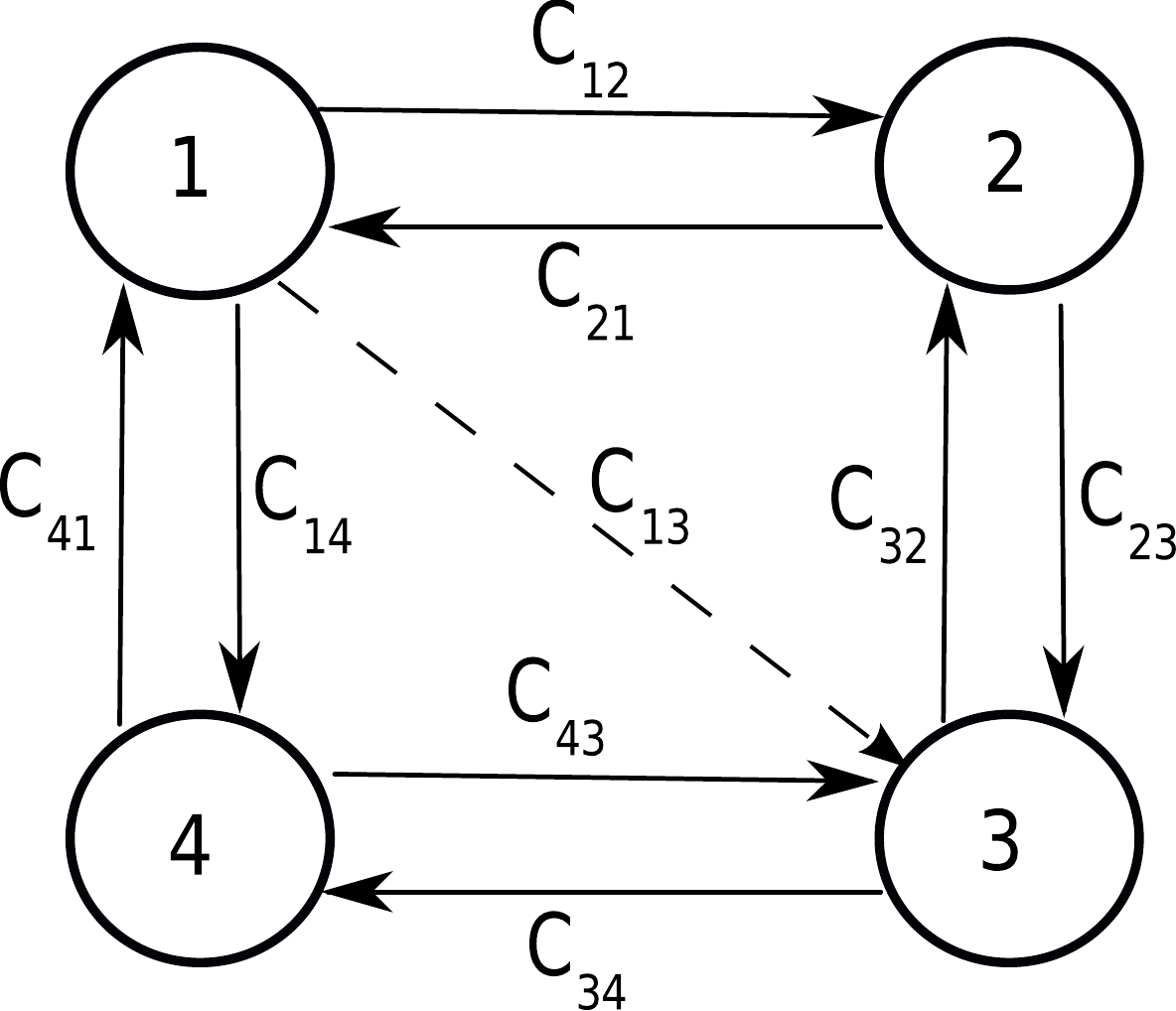}
	\label{fig:vc50}
}
\hspace{0.8cm}
\subfigure[Firewall-paths of the primary-conduits in (a).] 
{
	\includegraphics[scale=0.4]{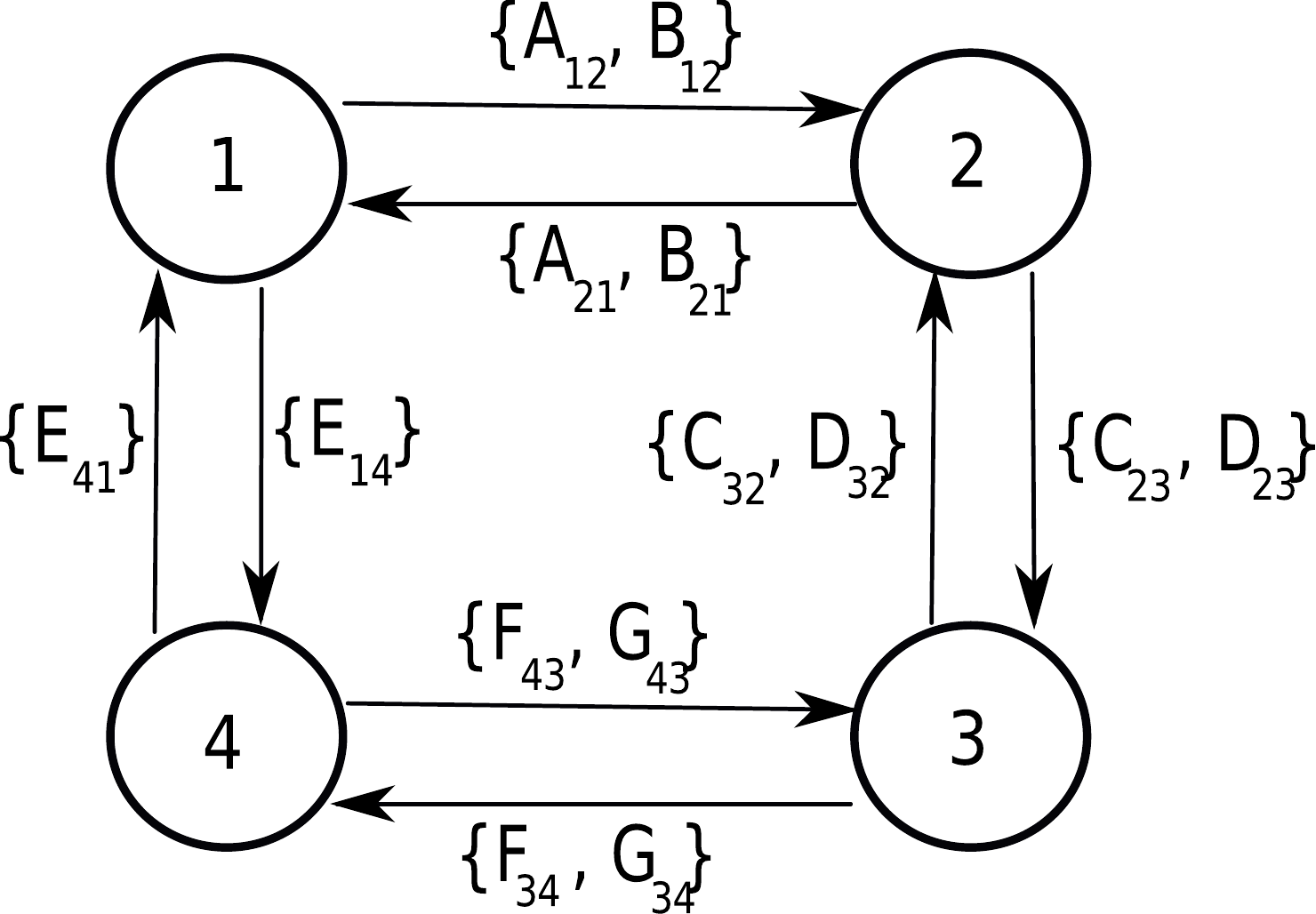}
	\label{fig:vcf51}
}
\caption{Zone-Conduit model depicting primary-conduits $C_{ij}$ and a secondary-conduit $C_{13}$ enabled by the transit zones 2 and 4 shown in (a). The firewall-paths of the primary-conduits are shown in (b).}
\label{fig:zc50}
\end{figure*}

But, the Zone-Conduit model is a logical representation, so, every firewall-path encoding in $F^*$ may not be valid.
We follow the design principles below to filter-out invalid paths:
\vspace{-1mm}
\begin{itemize}
\item{non-elementary paths in the Zone-Conduit model must be {\em excluded}.}
\vspace{-2mm}
\item{when traversing elementary Zone-Conduit paths, traffic packets shouldn't traverse a particular device interface more than once. 
So, for directed-firewalls $t_{ij},t_{jk} \in F$ the concatenation $t_{ij}t_{jk}$ is valid {\em iff} $h(t_{ij}) \neq h(t_{jk})$. 
We demonstrated the idea using the simple example in \autoref{fig:fwloop} earlier. }
\vspace{-2mm}
\item{traffic cannot flow through a non-transit zone. 
Typical zones this restriction applies to are the protected SCADA-Zones and Firewall-Zones.}
\end{itemize}

We have outlined the design principles for constructing a single firewall-path in a network.
In actuality, these principles and methods apply equally when constructing device-paths in other policy contexts.

We define single firewall-path concatenation based on the above design principles

\begin{defn}[Single firewall-path concatenation]
Take $a \in F^* \  \text{s.t.} \ a=t_{d_1d_2}t_{d_2d_3}t_{d_3d_4}...t_{d_nd_{n+1}} \ \text{where} $
$\ t_{d_id_j} \in F$. Also take $b \in F^* \ \text{s.t.} \ b=s_{g_1g_2}s_{g_2g_3}s_{g_3g_4}...s_{g_mg_{m+1}}$
$\ \text{where} \ s_{g_ig_j} \in F$. Then firewall-path concatenation from $F^* \times F^* \rightarrow F^*$ is \\\\
$ab= \begin{cases}
  t_{d_1d_2}...t^n_{d_nd_{n+1}}s_{g_1g_2}...s^m_{g_mg_{m+1}} \ \text{if} \ d_{n+1} = g_1 \\
  \ \ \text{and} \ g_i \neq d_j; \ \forall i,j, \ i >1 \\
  \ \ \text{and} \ h(t_{d_id_j}) \neq h(s_{g_kg_l}); \ \forall i,j,k,l \\  
  \phi, \ \text{otherwise}.
 \end{cases} $
 \label{v}
\end{defn}
\noindent and  $a \phi=\phi a= \phi ; \ \forall a \in F^*$

\smallskip
Concatenation in Definition 3 is a binary operation that constructs only elementary firewall-paths.

Consider two directed-firewalls $X,Y$ with policies $p_X$ and $p_Y$ that accepts packet sets $R$ and $T$.
The resultant policy of the concatenated firewall-path $XY$ has an action similar to that of sequential firewalls and can be denoted as 
$(p_X \otimes p_Y)(s) = p_{R \cap T}(s)$ where $s$ is a packet sequence.
The resultant policy action also depends on the policy context,. For instance, with traffic measurement policies, 
the outcome is similar to that of serial arbiters- \ie $p_X \otimes p_Y = p_{R \cup T}$.

We define a set of directed-firewall paths as a union of elements in $F^*$. 
Again, consider our directed-firewalls $X,Y$ from before.
If these described two distinct paths, then
the resultant policy of the path union $X \cup Y$ is that of parallel firewalls, \ie
$(p_X \oplus p_Y)(s) = p_{R \cup T}(s)$ where $s$ is a packet sequence. This result is also context dependent,
so, for traffic measurement policies, the outcome is similar to that of parallel arbiters, \ie
 $p_X \oplus p_Y = p_{R \cap T}$.

We also extend concatenation in Definition 2 to $S=\{\text{Power-set of } F^*\}$
\vspace{-1mm}
\begin{defn}[Multiple firewall-path concatenation]
Take $a,b \in S \  \text{s.t.} \ a=\{a_0,a_1...,a_x\}, b=\{b_0,b_1,...,b_y\} \ \text{where} \ a_i, b_j \in F^*$.
Then firewall-path concatenation from $S \times S \rightarrow S$ is given by

\smallskip
\indent \indent \indent \indent $ab= \{a_ib_j\}; \ \forall ij$  
 \label{v}
\end{defn}

\noindent Then, Definition 3 allows all possible sets of firewall paths
to be constructed from the union of elements in $S$ (\ie $\forall a,b \in S, \ a \cup b \in S$).
Then $(S,\cup,\cdot,\hat{0},\hat{1})$ is an idempotent semiring with

\smallskip
$\hat{0}=\phi; \ \text{empty set}$; and

$\hat{1}=\{\epsilon\}; \ \text{empty-string set}$ where
$\epsilon$ is the identity element of the concatenation operation.

\smallskip
\noindent It is important to note here that the properties of the operators $\cup$ and $\cdot$ actually dictate rules for firewall-path construction. 
So, when valid firewall-paths are built, they must be composed in a way that preserves these semantics. 
For instance, $\cup$ is commutative while $\cdot$ is not. So, the order of the directed-firewalls matter, when constructing a single firewall-path, but, are irrelevant
when constructing multiple paths.

Similarly, we can construct single and multiple device-paths for other policy contexts and
obtain the semiring result for $(S,\cup,\cdot,\hat{0},\hat{1})$ per security policies.

\vspace{-2mm}
\subsection{Mapping Policy to Arbiters}
\vspace{-2mm}
In the previous section we described how the semiring $(S,\cup,\cdot,\hat{0},\hat{1})$ constructs sets of device-paths between policy-graph endpoints.
The sequential (\ie $\otimes$) and parallel (\ie $\oplus$) policy-composition operators in \autoref{sec:abstractions} can now be used to construct the policies of these paths.
For instance,  assume we have to implement a high-level policy $p_{ij}$ on arbiters $q_{kl}$ that lie in the paths from $i$ to $j$.
All applicable device-paths from $i$ to $j$ can be constructed as per \autoref{sec:encoding} and given by
$S_{ij}=\{q_{a_1a_2}q_{a_2a_3}...q_{a_{n-1}a_n}, ....., q_{b_1b_2}q_{b_2b_3}...q_{b_{m-1}a_m}\}$.
Then, the high-level policy $p'_{ij}$ derived from the individual arbiter policies $p'_{kl}$ is
\begin{eqnarray}
  \label{eq:s2}
  p'_{ij}&=&(p'_{a_1a_2}\otimes p'_{a_2a_3}...\otimes p'_{a_{n-1}a_n}) \nonumber\\
  &&\oplus \  (....................................) \nonumber \\
  &&\oplus \ (p'_{b_1b_2} \otimes p'_{b_2b_3}...\otimes p'_{b_{m-1}b_m}).
\end{eqnarray}
\indent Mapping the intended high-level policy $p_{ij}$ to the arbiters is now a matter of finding $p'_{ij}$ for all arbiters such that
$p'_{ij}=p_{ij}$. But deriving such a mapping is non trivial because the meaning of $\oplus$ and $\otimes$ depends on the policy context.
For instance, with security policies $\oplus$ means {\em union} and $\otimes$ means {\em intersection}. So, a simple solution that supports {\em defence in depth} would be to implement the access-control policy $p_{ij}$ on every sequential firewall across all paths.\\
\indent For another instance, $\oplus$ means {\em summation} and $\otimes$ means {\em minimum} in QoS policies. So, a required bandwidth guarantee $p_{ij}$ can be split across multiple paths with sequential arbiters in each path guaranteeing only a portion of the total bandwidth.\\
\indent Similarly, with traffic measurement policies, $\oplus$ means {\em intersection} and $\otimes$ means {\em union}. So, each path must implement policy $p_{ij}$.
This can be achieved in multiple ways: (i) we could implement $p_{ij}$ on every sequential arbiter of a path for redundancy; or (ii) efficiently have only one arbiter per path implementing the policy; or
(iii) find an intermediate solution by partitioning $p_{ij}$ between sequential arbiters in a path.\\
\indent Irrespective of the policy context, the underlying requirement when mapping policy to network devices is to adhere to the semantics of \autoref{eq:s2}.
\vspace{-2mm}
\subsection{Computation of All Firewall-paths}
\label{sec:calc}
\vspace{-3mm}
We have managed to reduce the policy-to-device mapping problem to the semantics of \autoref{eq:s2} in the previous section.
We now describe how an algorithm can be developed to compute the device-paths of \autoref{eq:s2} efficiently.
Again, we demonstrate the idea using security policies and the Zone-Conduit model.

We can represent the firewall-paths of the primary-conduits using a {\em generalised Adjacency matrix $A$}. 
Here, $A(i,j)$ is the firewall-path of primary conduit $C_{ij} \in DC$. For our example in Figure \autoref{fig:vcf51}, $A=$
\smallskip

\begin{equation}
\begin{matrix}
  \begin{bmatrix}
  \overmat{$Z_1$}{\{\epsilon\}} & \overmat{$Z_2$}{\{A_{12},B_{12}\}} & \overmat{$Z_3$}{\phi} & \overmat{$Z_4$}{\{E_{14}\}} \\
  \{A_{21},B_{21}\} & \{\epsilon\} & \{C_{23},D_{23}\} & \phi \\
  \phi & \{C_{32},D_{32}\} & \{\epsilon\} & \{F_{34},G_{34}\} \\
  \{E_{41}\} & \phi & \{F_{43},G_{43}\} & \{\epsilon\} \\  
  \end{bmatrix}
\end{matrix}
\end{equation}

 Given such a matrix $A$, the solution to the problem of finding the set of all {\em valid}
 firewall-paths between zones is a matrix $A^*$ such that
 \begin{eqnarray}
  \label{eq:s3}
  A^*(i,j)&=&\{ \text{valid primary- and secondary-conduit} \nonumber \\
 &&\text{firewall-paths from zone i to j}\}
\end{eqnarray}

We developed the following right iteration algorithm\footnote{See appendix for proof.} to compute $A^*$, 
inspired by algorithms in meta-routing \cite{griffin2013}.
In doing so, we allow the node property: {\em traffic transitivity}, to be explicitly specified via a {\em zone-transitivity} matrix $T$.
\begin{eqnarray}
  \label{eq:s4}
 A^{<0>}&=&I \nonumber \\ 
 A^{<k+1>}&=&({A^{<k>} T} \ \cup \ I) A
\end{eqnarray}

The {\em zone-transitivity} matrix $T$ is defined as
\begin{equation}
  \label{eq:p}
  T(i,j)= 
  \begin{cases}
  \{\epsilon\} \ \text{if} \  i=j \ \text{and} \ \text{zone\_transitivity}(i)=1\\
  \phi, \ \text{otherwise}.
  \end{cases}
\end{equation} 

and $I$ is the multiplicative-identity matrix

\begin{equation}
  \label{eq:h}
  I(i,j)= 
  \begin{cases}
  \hat{1} \ \text{if} \  i=j \\
  \hat{0}, \ \text{otherwise}.
  \end{cases}
\end{equation} 

\smallskip
For bounded semirings we only iterate $n-1$ times to converge to $A^*$, where $n$ is the number of nodes in the Zone-Conduit model, \ie

\smallskip
\indent \indent  $A^*=A^{<n-1>}$

\smallskip
We have defined $T$, $A^*$ and the right-iteration algorithm for a security-policy context, but these can equally be defined for other policy contexts.

If we apply our algorithm in \autoref{eq:s4} to the example in Figure \autoref{fig:vcf51}, $n=4$, so, $A^*=A^{<3>}$ and let's {\em assume for simplicity that all zones are transitive}. Then

\begin{equation}
\begin{matrix}
  T=
  \begin{bmatrix}
  \{\epsilon\} & \phi & \phi & \phi \\
  \phi & \{\epsilon\} & \phi & \phi \\
  \phi & \phi & \{\epsilon\} & \phi \\
  \phi & \phi & \phi & \{\epsilon\} \\
 \end{bmatrix}
\end{matrix}
\end{equation}

and we see that $I=T$ in this instance.
\begin{eqnarray}
  \label{eq:s5}
 A^{<1>}&=&({A^{<0>} T} \ \cup \ I) A=({I T} \ \cup \ I) A \nonumber \\ 
 &=&(I \ \cup \ I)A=A \nonumber \\
 A^{<2>}&=&({AT} \ \cup \ I)A \nonumber
\end{eqnarray}

\smallskip
So, $A^*=A^{<3>}=({A^{<2>} T} \ \cup \ I) A=$

\begin{equation}
\begin{matrix}
  \begin{bmatrix}
  \{\epsilon\} & \eta & \kappa & \theta \\
  \mu & \{\epsilon\} & \nu & \beta \\
  \gamma & \lambda & \{\epsilon\} & \rho \\
  \xi & \delta & \zeta & \{\epsilon\} \\
 \end{bmatrix}
\end{matrix}
\end{equation}
where for instance
\begin{eqnarray}
  \label{eq:e}
 \kappa&=&\{A_{12}C_{23},A_{12}D_{23},B_{12}C_{23},B_{12}D_{23},E_{14}F_{43} \nonumber \\ 
 &&E_{14}G_{43}\} 
\end{eqnarray}

All valid firewall-paths from zone 1 to 3 are in $\kappa$.

\smallskip
Let's now {\em assume that zone \verb|4| is non-transitive}, then $zone\_transitivity(4)=0$, and by \autoref{eq:p}, 

\begin{equation}
\begin{matrix}
  T=
  \begin{bmatrix}
  \{\epsilon\} & \phi & \phi & \phi \\
  \phi & \{\epsilon\} & \phi & \phi \\
  \phi & \phi & \{\epsilon\} & \phi \\
  \phi & \phi & \phi & \phi \\
 \end{bmatrix}
\end{matrix}
\end{equation}

We re-calculate $A^*=$

\begin{equation}
\begin{matrix}
  \begin{bmatrix}
  \{\epsilon\} & \{A_{12},B_{12}\} & \eta & \theta \\
  \{A_{21},B_{21}\} & \{\epsilon\} & \{C_{23},D_{23}\} & \mu \\
  \gamma & \{C_{32},D_{32}\} & \{\epsilon\} & \rho \\
  \xi & \delta & \zeta & \{\epsilon\} \\
 \end{bmatrix}
\end{matrix}
\end{equation}

where for instance
\begin{eqnarray}
  \label{eq:s6}
 \eta=\{A_{12}C_{23},A_{12}D_{23},B_{12}C_{23},B_{12}D_{23}\}
\end{eqnarray}


\smallskip
The updated firewall-paths from 1 to 3 are in $\eta$. In comparison to \autoref{eq:e},
we see that the paths via zone 4 (\ie $E_{14}F_{43},E_{14}G_{43}$) have now been removed as the zone is no longer transitive.
$A^*$ can similarly be calculated for other policy contexts to determine all valid device-paths between endpoint pairs.

\smallskip
\noindent We will next describe our implementation of the algorithm in \autoref{eq:s4}. 
The implementation allows us to use these algebras to map policy to real network devices.
\vspace{-5mm}

\section{\uppercase{Implementation}}
\label{sec:imp}
\vspace{-2mm}
Our policy mapping system is depicted in \autoref{fig:map}, and we outline it's details below.
The system is currently implemented in Python, and allows mapping of high-level policies written in our
own policy specification language, to network devices. We will make the system open source in the near future.

\noindent \textbf{\em High-level security policy:} The topology-independent policy input file created using our high-level policy specification language.

\noindent \textbf{\em Network topology:} The input network topology described in the XML-based graph file format {\em GraphML} \cite{GraphML}.
The file holds information of all devices in the network and their interconnections. The crucial aspects are the details of the
topology near the policy arbiters (\eg firewalls).

\noindent \textbf{\em Map policy to network:} Compiles high-level policy to an Intermediate-Level (IL), generates the policy graph (\eg Zone-Conduit model) of the input network and computes $A^*$ as per \autoref{sec:calc} to map high-level policy to the devices (\eg firewalls) in the network.

\noindent \textbf{\em Policy-to-device map:} The primary output depicting policy breakdown by device, interface and direction.

\noindent \textbf{\em Verification output:} Secondary output specifically suited for verification (\eg policy errors).
\begin{figure}[t!]
\captionsetup{aboveskip=8pt}
\centerline{\includegraphics[scale=0.35]{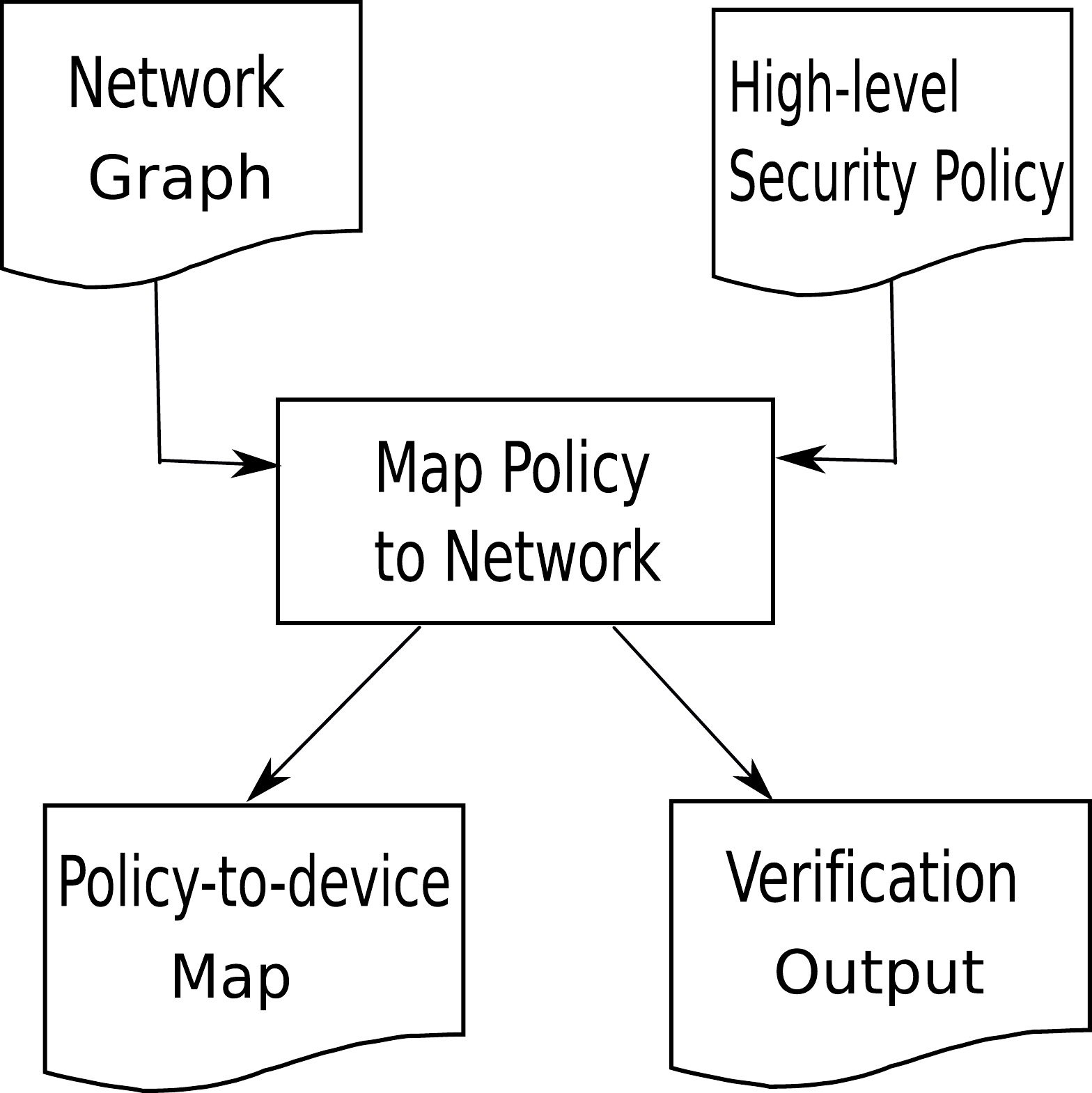}}
\caption{Policy to network-device mapping process.}
\label{fig:map}
\end{figure}

\smallskip
\noindent Our system implements a high-level policy on a network by coupling the policy to the topology instance. 
Details of the coupling steps are discussed in detail in our previous work in \cite{ranathunga2015V}.
We outline here, the steps taken to map policy to the network devices. Again, we illustrate our implementation 
using security policies and the Zone-Conduit model, but the system can likewise be used in other policy contexts.
\vspace{-2mm}
\subsection{Zone-transitivity Matrix and Adjacency Matrix Construction}
\vspace{-2mm}
Our system first builds the \textit{Zone-Firewall model}, containing the disjoint zones and their 
firewall interconnections using the input network topology. 
Additional Firewall-Zones, Abstract-Zones and Carrier-Zones \cite{ranathunga2015} are added to the model as required. 
The primary conduits are then defined to create the Zone-Conduit model. 

Next, the zone-transitivity capabilities are collated by compiling the high-level policy to IL policy.
Then the $n \times n$ zone-transitivity matrix ($T$) is constructed ($n$ is the number of zones in the network). 

Implicit mappings between the primary conduits and their firewall compositions are automatically created when generating the Zone-Conduit model. 
These mappings are leveraged to construct the adjacency matrix $A$ described in \autoref{sec:calc}.
\vspace{-2mm}
\subsection {$A^*$ Calculation}
\vspace{-2mm}
We compute $A^*$ using the right iteration algorithm in \autoref{eq:s4}.
Our implementation is given below (Algorithm 1). 
We consider a centralised implementation (given typically low $n$), but, it can also be 
distributed across multiple nodes performing parallel computations.

\begin{table*}[t]
\caption{SCADA case study summary adapted from \cite{ranathunga2015} ($*$ includes backup firewalls, $\#$ ACL rule allocation error).} 
\centering 
\begin{tabular}{c l c c c c c c c c c c} 
\hline\hline 
SUC & \multicolumn{1}{p{1cm}}{\centering Config. \\ date} & \multicolumn{1}{p{0.6cm}}{\centering Fire-walls$^*$} &  \multicolumn{1}{p{0.65cm}}{\centering Zones}  & \multicolumn{1}{p{1cm}}{\centering Conduits} &  \multicolumn{1}{p{0.8cm}}{\centering Max. \\ hosts} &
\multicolumn{1}{p{0.55cm}}{\centering ACLs} & \multicolumn{1}{p{1.25cm}}{\centering Average rules \\ per ACL } & \multicolumn{1}{p{1.1cm}}{\centering Incorrect \\ firewall$^{\#}$ } & \multicolumn{1}{p{1.1cm}}{\centering Incorrect \\ interface$^{\#}$} & \multicolumn{1}{p{1.1cm}}{\centering Incorrect \\ direction$^{\#}$} & \multicolumn{1}{p{1.1cm}}{\centering Runtime \\ (seconds)}  \\ [0.5ex] 
\hline 
1 & Sep 2011 & 3  & 7 &11 & 67580 & 8 & 237 & 15  & 13 & 19 & 40 \\ 
2 & Aug 2011 & 6 & 21 & 81 & 2794 & 12 & 16 & 3 & 2 & 5 & 70\\
3 & Oct 2011 &  4 & 10 & 17 & 886 & 8 &  6 & 2 & 1 & 4 & 43\\ 
4 & Mar 2011 &  3 & 9 & 16 & 2038 & 3 & 80 & 5 & 12 & 13 & 61\\
5 & Apr 2015 &  3 & 12 & 19 & 2664 & 12 & 677 & 15 & 8 & 26 & 47\\
6 & Apr 2015 &  3 & 13 & 21 & 3562 & 8 & 1034 & 21 & 15 & 19 & 63\\
7 & Jul 2015 &  6  & 15 & 22 & 3810 &  17 & 724 &  9 & 5 & 17 & 49\\
\hline 
\end{tabular}
\label{table:suc-summary} 
\end{table*}

\begin{algorithm}[t]
\caption{Right-iteration algorithm to compute $A^*$.
The additive and multiplicative operators of Semiring $S$ are $\cup, \cdot$ (\autoref{sec:encoding}).
The additive and multiplicative identities of $S$ are $\hat{0}, \hat{1}$. 
The multiplicative-identity matrix and the zone-transitivity matrix are $I$ and $T$.
The set of zones in the network is $Z$.}
\label{alg:right_iter}
\begin{algorithmic}[1]
\Procedure{RightIteration}{$(\cup, \cdot, \hat{0}, \hat{1},A)$}
\State $R_0 \gets \text{I}$
\State $k \gets \text{1}$

\While{$k \leq |Z| -1$}
	\ForEach{$i \in Z$}
		\ForEach{$j \in Z$}
		      \If{i=j} 
		      	  \State $R_{k+1}(i,j) \gets I(i,j)$
		      \Else
		          \State $R_{k+1}(i,j) \gets \hat{0}$
		          \ForEach{$q \in Z$}
		               \State $s \gets R_k(i,q) \cdot T(q,j)$
		               \State $R_{k+1}(i,j) \gets R_{k+1}(i,j) \cup s$		               
		               \State $R_{k+1}(i,j) \gets R_{k+1}(i,j) \cup I(i,q)$
		               \State $R_{k+1}(i,j) \gets R_{k+1}(i,j) \cdot A(q,j)$
		          \EndFor
		      \EndIf
		\EndFor
	\EndFor
	\State $k \gets \text{k+1}$
 \EndWhile
\EndProcedure
\end{algorithmic}
\end{algorithm}
Lines 2,3,10 initialise $R_0$, $k$ and $R_{k+1}(i,j)$ while lines 5,6,11 range over the set of zones $Z$ in the network.
At each step $k$, we store the result $A_{k+1}$ in $R_{k+1}$.

We can see that the above algorithm has time complexity $O(n^4)$. This value suggests that the algorithm performance decreases exponentially as the number of zones in the network ($n$) increases.
So, specifying policy per individual host (\ie zone size=1 so, large $n$) in {\em top-down} configuration makes policy mapping {\em extremely inefficient}. 
It proves more efficient to create network groups (\eg by subnet) and specify policy between them, so, a reasonably low value can be maintained for $n$.
For instance, we will see in \autoref{sec:example} that in a SCADA network typically $n< 25$. 
Our mapping algorithm runs on average in 53 seconds 
on a standard desktop computer for each of these networks.

In calculating $A^*$, we considered all valid paths between two zones (hence time complexity of $O(n^4)$). The decision allows us to permit or block traffic along all possible communication paths between two zones, providing redundancy and {\em defence-in-depth} in the network. 
We also map policy uniformly to every firewall along a single-path, further boosting defence-in-depth.
These decisions collectively create a more robust defence against cyber attacks. 

However, in other policy contexts, it may be useful to select a {\em subset} of all valid paths or even just a single path (\eg shortest path) instead.
Doing so, could improve the time complexity of \autoref{eq:s4} (\eg selecting shortest paths would yield $O(n^3)$).
This path pruning can be done, for instance, by incorporating a sparse matrix in \autoref{eq:s4}.
It may also be efficient to apply policy only on edge arbiters of a path, in some contexts (\eg when enabling traffic measurements).
\vspace{-7mm}

\section{\uppercase{A series of case studies}}
\label{sec:example}
\vspace{-3mm}
We now demonstrate the use of our policy-mapping algebras,
through seven real SCADA-firewall configuration case studies summarised in
\autoref{table:suc-summary}.  The data was provided by the authors of
\cite{ranathunga2015}.

The seven Systems Under Consideration (SUCs), involve various
firewall architectures and models. We use them to demonstrate several
properties, most notably that the computational complexity of our policy-to-device mapping algorithm
is tractable.

An important feature depicted in \autoref{table:suc-summary} is the {\em number of security zones} in each network.
This number is {\em small} (\ie $\leq 21$) relative to the maximum (potential) number of hosts per network (\ie $\leq 67580$).

This is to be expected, a zone groups a set of hosts or subnets with identical policies.
If every host had a distinct policy then a large number of firewalls would be needed to enforce a real separation between the hosts, making it impractical.
By grouping hosts into zones, we reduce policy complexity, so their specification becomes easier and less error-prone.\\
\indent The typically small zone-count, makes Algorithm 1 (\autoref{sec:imp}) computationally feasible.
For instance, the worse-case computational complexity involving zones is $O(21^4)$. In comparison, 
the worse-case computational complexity involving hosts is $O(67580^4)$! \\
\indent The average time taken by our system to map high-level policy to network devices is 53 seconds per case study.
This is when the system is run on a standard desktop computer
(\eg Intel Core CPU 2.7-GHz computer with 8GB of RAM running Mac OS X).\\
\indent We identified incorrectly mapped (\ie assigned) ACL rules in each case study by parsing the firewall configurations as per \cite{ranathunga2015}.
These errors were then classified into three groups: {\em incorrect-firewall, incorrect-interface and incorrect-direction} errors (\autoref{table:suc-summary}).
{\em Incorrect-firewall} errors are ACL rules that are assigned to the wrong firewall to begin with, \ie the desired traffic filtering could not be achieved by placing the ACL rule in any of that firewall's interfaces. 
{\em Incorrect-interface} errors are ACL rules that are assigned to the correct firewall but to the wrong firewall-interface, \ie the desired traffic filtering could not be achieved by assigning the rule inbound or outbound of that firewall-interface.
{\em Incorrect-direction} errors comprise of ACL rules that are assigned to the correct firewall and firewall-interface, but in the wrong direction (\eg outbound instead of inbound).\\
\indent As (\autoref{table:suc-summary}) suggests, on average there were 10 ACL rules allocated to the wrong firewall, 8 rules allocated to the wrong firewall-interface and 15 rules allocated in the wrong interface-direction, per SCADA case study.
Hence, the intended security policy was {\em not correctly implemented} in any of these mission-critical networks! Needless to say, what chance do we have in preventing
financial loss or moreover, the loss of human lives when our critical infrastructure have incorrectly deployed security policies to begin with?\\
\indent We automatically mapped the high-level policy in each case to it's network using our system. 
The generated policy rules were then checked for incorrect allocations.
There were {\em zero incorrectly allocated policy rules}, when the policy was mapped to the firewalls using our algebras!
So, the intended high-level security policies were now correctly deployed to the firewalls.
By mapping policy to the correct set of firewalls, we reduce vulnerability of these SCADA networks to cyber attack
and prevent catastrophic outcomes that result from such attacks.
\vspace{-6mm}

\section{\uppercase{Conclusions}}
\label{sec:conclusion}
\vspace{-2mm}
There are various obstacles that hinder the precise mapping of policies to network devices.
Most prominent is the lack of decoupling between policy and network which makes
policy sensitive to network-intricacies and vendor changes.
To compound the problem, policies can also have complex
semantics including node and link properties.

Our research addresses these challenges and proposes
a mathematical foundation for mapping policies to network-devices.
We use it to deploy real-world security policies to network firewalls
provably correctly, so, that administrators can be confident of the protection provided by their policies for their networks.
\vspace{-7mm}





\nocite{howe1996}
\nocite{foster2010}
\nocite{reich2013}
\nocite{soule2014}
\nocite{prakash2015}
\nocite{bartal1999}
\nocite{griffin2013}
\nocite{GraphML}
\nocite{anderson2014}
\nocite{cisco2014b}
\nocite{gutz2012}
\nocite{griffin2013}

\nocite{patrick1995}
\nocite{byres2005}
\nocite{isa2007}
\nocite{ranathunga2015}
\nocite{ranathunga2015T}
\nocite{ranathunga2015V}
\nocite{ranathunga2015P}
\nocite{harary1960}
\nocite{yuan2006}
\nocite{al2004}
\nocite{guttman2005}
\nocite{twidle2009}
\nocite{zhao2008}
\nocite{avelsgaard1990}
\nocite{Johnson:2005:NC:1077464.1077476}

\addcontentsline{toc}{section}{List of Acronyms}
\begin{acronym}[TDMA]
        \setlength{\itemsep}{-\parsep}
        \acro{IANA}{Internet Assigned Numbers Authority}
        \acro{OOP}{Object Oriented Programming}
        \acro{BNF}{Backus-Naur Form} 
        \acro{MS-SQL}{Microsoft SQL Server}
        \acro{IL}{Intermediate Level}
        \acro{SCADA}{Supervisory Control and Data Acquisition}
        \acro{COTS}{Commodity Off-The-Shelf}
        \acro{TCP}{Transmission Control Protocol}
        \acro{IP}{Internet Protocol}
        \acro{ISA}{International Society for Automation}
        \acro{ASA}{Adaptive Security Appliance}
        \acro{PIX}{Private Internet eXchange}
        \acro{IOS}{IOS}
        \acro{WAN}{Wide Area Network}
        \acro{CLI}{Command Line Interface}
        \acro{DoS}{Denial of Service}         
        \acro{ACLs}{Access Control Lists}
        \acro{ACE}{Access Control Entry}
        \acro{ACEs}{Access Control Entries}
        \acro{CN}{Corporate Network}
        \acro{CZ}{Corporate-Zone}
        \acro{SZ}{SCADA-Zone}
        \acro{MZ}{Management-Zone}
        \acro{FWZ}{Firewall-Zone}
        \acro{IZ}{Internet-Zone}  
        \acro{IDS}{Intrusion Detection Systems} 
        \acro{DMZ}{Demilitarised-Zone}
        \acro{UDP}{User Datagram Protocol}
        \acro{DNS}{Domain Name System}
        \acro{HTTP}{Hypertext Transfer Protocol}
        \acro{HTTPS}{Hypertext Transfer Protocol Secure}
        \acro{FTP}{File Transfer Protocol}
        \acro{NTP}{Network Time Protocol}
        \acro{SSH}{Secure Shell Protocol}
        \acro{SMTP}{Simple Mail Transfer Protocol}
        \acro{SNMP}{Simple Network Management Protocol}
        \acro{ICMP}{Internet Control Message Protocol}
        \acro{DCOM}{Distributed Component Object Model}        
        \acro{NAT}{Network Address Translation}
        \acro{VPN}{Virtual Private Network}
        \acro{VPNs}{Virtual Private Networks}
        \acro{ANSI}{American National Standards Institute}
        \acro{OSI}{Open System Interconnection}
        \acro{SUC}{System Under Consideration}
        \acro{SUCs}{Systems Under Consideration}
        \acro{VM}{Virtual Machine}
        \acro{ACL}{Access Control List}
        \acro{UML}{User Mode Linux}
        \acro{LoC}{Lines of Code}
        \acro{PLCs}{Programmable Logic Controllers}
        \acro{IL}{Intermediate Language }
        \acro{SDN}{Software Defined Networking}
        
\end{acronym}

\bibliographystyle{apalike}
{\small
\bibliography{myBibliography}}

\begin{thebibliography}{}

\bibitem[Anderson et~al., 2014]{anderson2014}
Anderson, C.~J., Foster, N., Guha, A., Jeannin, J.-B., Kozen, D., Schlesinger,
  C., and Walker, D. (2014).
\newblock Netkat: Semantic foundations for networks.
\newblock {\em ACM SIGPLAN Notices}, 49(1):113--126.

\bibitem[Anonymous, 2015]{ranathunga2015V}
Anonymous (2015).
\newblock Forest{F}irewalls: Getting firewall configuration right in critical
  networks, \url{http://tinyurl.com/jdsfewk}.

\bibitem[Anonymous, 2016]{ranathunga2015P}
Anonymous (2016).
\newblock {Malachite:} {F}irewall policy comparison,
  \url{http://tinyurl.com/o2ke3py}.

\bibitem[ANSI/ISA-62443-1-1, 2007]{isa2007}
ANSI/ISA-62443-1-1 (2007).
\newblock Security for industrial automation and control systems part 1-1:
  Terminology, concepts, and models.

\bibitem[Avelsgaard, 1990]{avelsgaard1990}
Avelsgaard, C. (1990).
\newblock {\em Foundations for Advanced Mathematics}.
\newblock Scott, Foresman, Brown Higher Education.

\bibitem[Bartal et~al., 2004]{bartal1999}
Bartal, Y., Mayer, A., Nissim, K., and Wool, A. (2004).
\newblock Firmato: A novel firewall management toolkit.
\newblock {\em ACM TOCS}, 22(4):381--420.

\bibitem[Billingsley, 1995]{patrick1995}
Billingsley, P. (1995).
\newblock Probability and measure.
\newblock {\em A Wiley-Interscience Publication,Wiley \& Sons, New York}.

\bibitem[Byres et~al., 2005]{byres2005}
Byres, E., Karsch, J., and Carter, J. (2005).
\newblock {NISCC} good practice guide on firewall deployment for {SCADA} and
  process control networks.
\newblock {\em NISCC}.

\bibitem[{Cisco Systems}, 2014]{cisco2014b}
{Cisco Systems} (2014).
\newblock {\em {C}isco {V}irtual {S}ecurity {G}ateway for {N}exus 1000{V}
  Series Switch Configuration Guide}.
\newblock Cisco Systems, Inc., 170 West Tasman Drive, San Jose, CA 95134-1706,
  USA.

\bibitem[Dynerowicz and Griffin, 2013]{griffin2013}
Dynerowicz, S. and Griffin, T.~G. (2013).
\newblock {O}n the forwarding paths produced by internet routing algorithms.
\newblock In {\em ICNP}, pages 1--10.

\bibitem[Foster et~al., 2010]{foster2010}
Foster, N., Freedman, M.~J., Harrison, R., Rexford, J., Meola, M.~L., and
  Walker, D. (2010).
\newblock Frenetic: a high-level language for openflow networks.
\newblock In {\em Proceedings of the Workshop on Programmable Routers for
  Extensible Services of Tomorrow}, page~6. ACM.

\bibitem[{Graph Steering Committee}, 2003]{GraphML}
{Graph Steering Committee} (2003).
\newblock {GraphML}.

\bibitem[Guttman and Herzog, 2005]{guttman2005}
Guttman, J.~D. and Herzog, A.~L. (2005).
\newblock Rigorous automated network security management.
\newblock {\em IJIS}, 4:29--48.

\bibitem[Gutz et~al., 2012]{gutz2012}
Gutz, S., Story, A., Schlesinger, C., and Foster, N. (2012).
\newblock Splendid isolation: {A} slice abstraction for software-defined
  networks.
\newblock In {\em ACM HotSDN}, pages 79--84.

\bibitem[Harary and Norman, 1960]{harary1960}
Harary, F. and Norman, R.~Z. (1960).
\newblock Some properties of line digraphs.
\newblock {\em Rendiconti del Circolo Matematico di Palermo}, 9(2):161--168.

\bibitem[Howe, 1996]{howe1996}
Howe, C.~D. (1996).
\newblock {\em {W}hat's Beyond Firewalls?}
\newblock Forrester Research, Incorporated.

\bibitem[Johnson, 2005]{Johnson:2005:NC:1077464.1077476}
Johnson, D.~S. (2005).
\newblock The {NP}-completeness column.
\newblock {\em ACM Transactions on Algorithms}, 1(1):160--176.

\bibitem[Prakash et~al., 2015]{prakash2015}
Prakash, C., Lee, J., Turner, Y., Kang, J.-M., Akella, A., Banerjee, S., Clark,
  C., Ma, Y., Sharma, P., and Zhang, Y. (2015).
\newblock {PGA}: {U}sing graphs to express and automatically reconcile network
  policies.
\newblock In {\em ACM SIGCOMM}, pages 29--42.

\bibitem[Ranathunga et~al., 2015a]{ranathunga2015T}
Ranathunga, D., Roughan, M., Kernick, P., and Falkner, N. (2015a).
\newblock {T}owards standardising firewall reporting.
\newblock In {\em 1st Workshop on the Security of Cyber Physical Systems
  (WOS-CPS)}. LNCS.

\bibitem[Ranathunga et~al., 2015b]{ranathunga2015}
Ranathunga, D., Roughan, M., Kernick, P., Falkner, N., and Nguyen, H. (2015b).
\newblock Identifying the missing aspects of the {ANSI/ISA} best practices for
  security policy.
\newblock In {\em 1st ACM Workshop on Cyber-Physical System Security (CPSS)},
  pages 37--48.

\bibitem[Reich et~al., 2013]{reich2013}
Reich, J., Monsanto, C., Foster, N., Rexford, J., and Walker, D. (2013).
\newblock Modular {SDN} programming with pyretic.
\newblock {\em Technical Report of USENIX}.

\bibitem[Soul{\'e} et~al., 2014]{soule2014}
Soul{\'e}, R., Basu, S., Marandi, P.~J., Pedone, F., Kleinberg, R., Sirer,
  E.~G., and Foster, N. (2014).
\newblock Merlin: A language for provisioning network resources.
\newblock In {\em Proceedings of the 10th ACM International on Conference on
  emerging Networking Experiments and Technologies}, pages 213--226.

\bibitem[Twidle et~al., 2009]{twidle2009}
Twidle, K., Dulay, N., Lupu, E., and Sloman, M. (2009).
\newblock Ponder2: A policy system for autonomous pervasive environments.
\newblock In {\em ICAS'09}, pages 330--335.

\bibitem[Yuan et~al., 2006]{yuan2006}
Yuan, L., Chen, H., Mai, J., Chuah, C.-N., Su, Z., and Mohapatra, P. (2006).
\newblock {F}{I}{R}{E}{M}{A}{N}: A toolkit for firewall modeling and analysis.
\newblock In {\em IEEE SSP}, pages 15--213.

\bibitem[Zhao et~al., 2008]{zhao2008}
Zhao, H., Lobo, J., and Bellovin, S.~M. (2008).
\newblock An algebra for integration and analysis of {P}onder2 policies.
\newblock In {\em POLICY'08}, pages 74--77.

\end{thebibliography}


\vspace{-7mm}
\section*{Appendix}
\vspace{-2mm}
\begin{thm}[$A^*$ calculation]
  \label{thm:alg}  
  $A^*$ can be calculated using the right iteration algorithm
  
  \smallskip
   $A^{<k+1>}=(A{^{<k>} T} \ \cup \ I) A$ where $A^{<0>}=I$.
\end{thm}

$T$ and $I$ are the Zone-transitivity matrix and the multiplicative-identity matrix (of semiring $(S,\cup,\cdot,\hat{0},\hat{1})$) respectively.
\begin{proof}

Let's check that the result holds for $k=0$ (\ie valid firewall paths between zones up to $(0+1)$ hop).
\begin{eqnarray}
  \label{eq:s10}
  LHS&=&A^{<0+1>}=A^{<1>} \nonumber \\
  RHS&=&( A{^{<0>} T} \ \cup \ I) A=( {I T} \ \cup \ I) A \nonumber \\
 &=&(T \ \cup \ I) A=I A=A \nonumber
\end{eqnarray}
\indent $A^{<1>}=A$ is true since all valid single-hop firewall paths are represented by $A$.

Let's assume the result holds when $k=n$, \ie
 $A^{<n+1>}=( (A{^{<n>} T}) \cup \ I) A$.

So, $A^{<n+1>}$ holds all valid firewall paths between the zones of up to $(n+1)$ hops.
Then $A^{<n+1>}(i,j)$ represents all valid firewall paths of up to $(n+1)$ hops, between zones $i$ and $j$ (see \autoref{fig:vc}). 
So, valid firewall paths up to $(n+2)$ hops from $i$ to $k$ via $j$ are given by
\begin{equation}
(t_{ik})_j=A^{<n+1>}(i,j) A(j,k)
\end{equation}

\begin{figure}[h!]
\captionsetup{aboveskip=8pt}
\centerline{\includegraphics[scale=0.32]{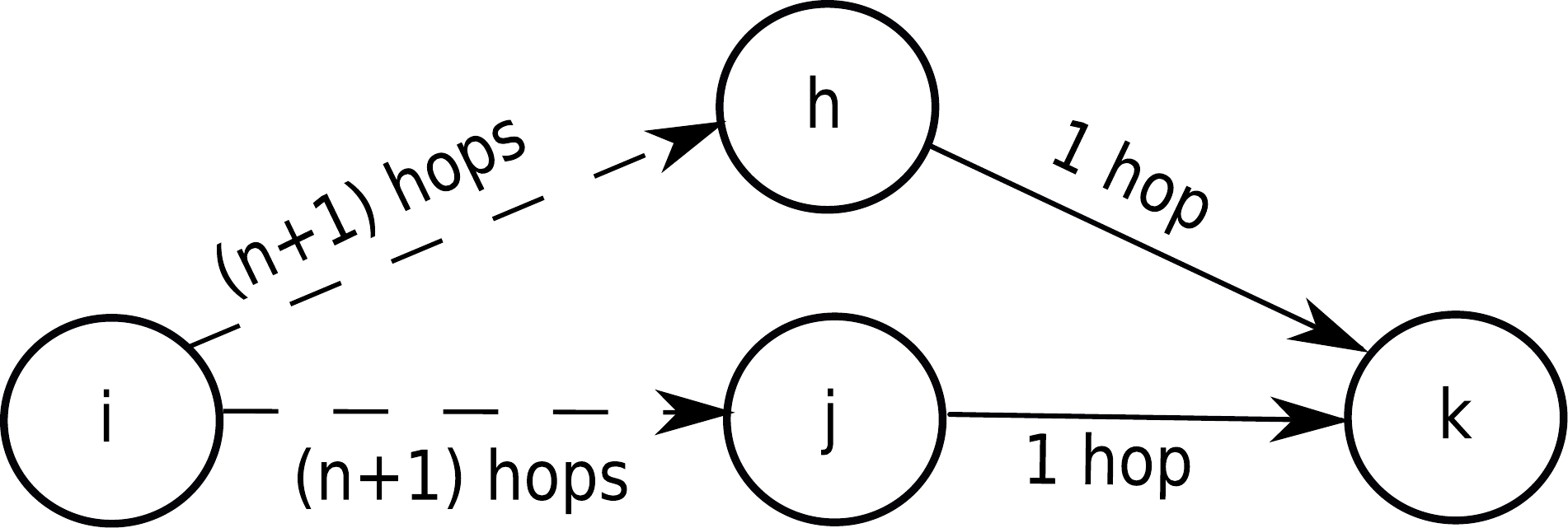}}
\caption{.}
\label{fig:vc}
\end{figure}

But zone $j$ may not be transitive, so we modify the above equation to account for this
\begin{equation}
(t_{ik})_j=A^{<n+1>}(i,j) T(j,j) A(j,k)
\end{equation}

But, when $i=j$ the above equation yields $\phi$, whereas we expect $\{\epsilon\}$ to be the self-firewall path.
Again, we update the equation
\begin{equation}
(t_{ik})_j=\Big[ \ A^{<n+1>}(i,j) T(j,j)  \ \cup \ I(i,j) \ \Big] A(j,k)
\end{equation}

Then, all valid firewall paths up to (n+2) hops from $i$ to $k$ (via $j,h...etc.$) are given by
\begin{eqnarray}
  \label{eq:s10}
  A^{<n+2>}(i,k)=&& \nonumber \\
   \bigcup_{\forall j} \Big[ \ A^{<n+1>}(i,j) T(j,j)\cup \ I(i,j) \ \Big] A(j,k); \ \forall i,k
\end{eqnarray}

We simplify the equation further
\begin{eqnarray}
  \label{eq:s10}
  RHS&=&\bigcup_{\forall j} \Big[ \ A^{<n+1>}(i,j) T(j,j)  A(j,k) \Big] \ \cup \nonumber \\
     &&\bigcup_{\forall j} I(i,j) A(j,k) \nonumber
\end{eqnarray}
\vspace{-2mm}
\begin{eqnarray}
  \label{eq:s10}
  \text{But,} \ A^{<n+1>}(i,j) T(j,j) = \bigcup_{\forall s}\Big[A^{<n+1>}(i,s) T (s,j)\Big] \nonumber
\end{eqnarray}
\vspace{-2mm}
\begin{eqnarray}
  \label{eq:s10}
  \text{So,} \ &&\bigcup_{\forall j}\Big[A^{<n+1>}(i,j) T(j,j)  A(j,k)\Big]  \nonumber \\
  &=&\bigcup_{\forall j}\ \Big[A^{<n+1>} T\Big](i,j) A(j,k) \nonumber \\
  &=&\Big[A^{<n+1>} T A\Big](i,k) \nonumber
\end{eqnarray}

\smallskip
Also, $\ \bigcup_{\forall j}\Big[I(i,j) A(j,k)\Big]=\Big[I A\Big](i,k)$
\begin{eqnarray}
  \label{eq:s10}
  \text{Hence,} \ A^{<n+2>}(i,k)&=&\Big[A^{<n+1>} TA\Big](i,k) \ \cup \nonumber \\
  && \Big[I A\Big](i,k); \ \forall i,k \nonumber
\end{eqnarray}

We can generalise the above to the matrices
\begin{eqnarray}
  \label{eq:s10}
  A^{<n+2>}&=&A^{<n+1>} T A \ \cup \ I A \nonumber \\
  &=&(A{^{<n+1>} T} \ \cup \ I) A. \nonumber
\end{eqnarray}
\end{proof}

\vfill
\end{document}